\documentclass[bjps,authoryear,preprint]{imsart}

\usepackage{amsthm,amsmath,natbib,amsfonts,graphicx,float,dsfont,amssymb,booktabs}
\RequirePackage[colorlinks,citecolor=blue,urlcolor=blue]{hyperref}

\startlocaldefs

\newtheorem{theorem}{Theorem}
\newtheorem{Definition}{Definition}
\newtheorem{Lemma}{Lemma}

\newtheorem{Proposition}{Proposition}

\DeclareMathOperator{\re}{\mathbb{R}}

\newcommand{\abs}[1]{\left|#1\right|}

\newcommand{\ind}{\mathds{1}}
\newcommand{\ee}{\mathrm{e}}
\newcommand{\Prob}{\mathbb{P}}

\DeclareMathOperator{\simdist}{\stackrel{\mathcal{D}}{\sim}}
\DeclareMathOperator{\yo}{y}

\DeclareMathOperator{\likeli}{\mathcal{L}}
\DeclareMathOperator{\convdist}{\stackrel{\mathcal{D}}{\rightarrow}}
\newcommand{\za}[1]{\stackrel{a}{#1}}
\newcommand{\zb}[1]{\stackrel{b}{#1}}
\newcommand{\zc}[1]{\stackrel{c}{#1}}
\newcommand{\zd}[1]{\stackrel{d}{#1}}
\newcommand{\ze}[1]{\stackrel{e}{#1}}

\endlocaldefs

\begin{document}

\begin{frontmatter}

\title{Bayesian Robustness to Outliers in Linear Regression and Ratio Estimation}

\runtitle{Bayesian Robustness to Outliers in Linear Regression and Ratio Estimation}

\begin{aug}
\author{\fnms{Alain} \snm{Desgagn\'{e}}\thanksref{uqam}\ead[label=e3]{desgagne.alain@uqam.ca}}
\and
\author{\fnms{Philippe} \snm{Gagnon}\thanksref{udem}\ead[label=e1]{gagnonp@dms.umontreal.ca}}

\affiliation{Universit\'{e} du Qu\'{e}bec \`{a} Montr\'{e}al\thanksmark{uqam} and Universit\'{e} de Montr\'{e}al\thanksmark{udem}}

\address{D\'{e}partement de math\'{e}matiques\\
         Universit\'{e} du Qu\'{e}bec \`{a} Montr\'{e}al\\
       C.P. 8888, Succursale Centre-ville\\
        Montr\'{e}al, QC, H3C 3P8, Canada \\
\printead{e3}}

\address{D\'{e}partement de math\'{e}matiques et de statistique \cr
Universit\'{e} de Montr\'{e}al \cr
C.P. 6128, Succursale Centre-ville \cr
Montr\'{e}al, QC, H3C 3J7, Canada \cr
\printead{e1}}

\runauthor{Desgagn\'{e} A. and Gagnon P.}

\end{aug}

\begin{abstract}
Whole robustness is a nice property to have for statistical models. It implies that the impact of outliers gradually vanishes as they approach plus or minus infinity. So far, the Bayesian literature provides results that ensure whole robustness for the location-scale model. In this paper, we make two contributions. First, we generalise the results to attain whole robustness in simple linear regression through the origin, which is a necessary step towards results for general linear regression models. We allow the variance of the error term to depend on the explanatory variable. This flexibility leads to the second contribution: we provide a simple Bayesian approach to robustly estimate finite population means and ratios. The strategy to attain whole robustness is simple since it lies in replacing the traditional normal assumption on the error term by a super heavy-tailed distribution assumption. As a result, users can estimate the parameters as usual, using the posterior distribution.
\end{abstract}



\end{frontmatter}

\section{Introduction}\label{intro}

Conflicting sources of information may contaminate the inference arising from statistical analysis. The conflicting information may come from outliers and also prior misidentification. In this paper, we focus on robustness with respect to outliers in a Bayesian simple linear regression model through the origin. We say that a conflict occurs when a group of observations produces a rather different inference than that proposed by the bulk of the data and the prior. Light-tailed distribution assumptions on the error term can lead to an undesirable compromise where the posterior distribution concentrates on an area that is not supported by any source of information. We believe that the appropriate way to address the problem is to limit the influence of outliers in order to obtain conclusions consistent with the majority of the observations.

\cite{1968tiao119} were the first to introduce a robust Bayesian linear regression model. They proposed to assume that the distribution of the error term is a mixture of two normals with one component for the nonoutliers and the other one, with a larger variance, for the outliers. This approach has been generalised by \cite{1984west431} who modelled errors with heavy-tailed distributions constructed as scale mixtures of normals, which include the Student distribution.
More recently, \cite{2009Pena2196} introduced a different robust Bayesian method where each observation has a weight decreasing with the distance between this observation and most of the data. They proved that the Kullback-Leibler divergence from the posterior arising from the nonoutliers only to the posterior arising from the sample containing outliers is bounded.

So far, the literature only provides solutions to attain whole robustness for the estimation of the slope
in the model of regression through the origin (e.g.\ if we assume that the error term has a Student distribution instead of a normal, see the results of \cite{andrade2011bayesian} in a context of location-scale model). However, only partial robustness is reached for the estimation of the scale parameter of the error term. Partial robustness means that the outliers have a significant but limited influence on the inference, as the conflict grows infinitely.
   In this  paper, we go  a  step further: we attain whole robustness to outliers for both the slope and scale parameters, in the sense that the impact of outliers gradually vanishes as they approach plus or minus infinity. To achieve this, we generalise the results of \cite{desgagne2015robustness}, which ensure whole robustness for both parameters of the location-scale model simultaneously, to the simple linear regression model through the origin. Our work is thus aligned with the \textit{theory of conflict resolution in Bayesian statistics}, as described by \cite{o2012bayesian} in their extensive literature review on that topic.

  The strategy to attain whole robustness for all parameters is, instead of assuming the traditional normality of the errors in the model, to assume that they have a super heavy-tailed distribution. The general model (with no specific distribution assumption on the error term) is described in Section~\ref{sec-model}. The class of super heavy-tailed distributions that we consider, which are log-regularly varying distributions, is presented in Section~\ref{sec-log-regularly}. When assuming a super heavy-tailed distribution on the error term, the resulting model is characterised by its built-in robustness that resolves conflicts in a sensitive and automatic way, as stated in our robustness results given in Section~\ref{sec-conflict}. The main result is the convergence of the posterior distribution towards the posterior arising from the nonoutliers only, when the outliers approach plus or minus infinity. Although our results are Bayesian analysis-oriented, they reach beyond this paradigm through the robustness of the likelihood function, and therefore, of both slope and scale maximum likelihood parameter estimation. These are the results that ensure that whole robustness is reached for the considered model.

  We believe our work will eventually lead to whole robustness results for the estimation of the parameters of the usual multiple linear regression model, which will in turn allow to introduce Bayesian robust ANOVA and t-test procedures. In fact, a preliminary numerical investigation suggests that similar results to those presented in this paper hold for multiple linear regressions. However, precise conditions and results will need to be specified. This can be achieved by the (non-trivial) extension of the proof presented in the following for the simple linear regression through the origin.

  In addition to representing a crucial step towards whole robustness for the more general case of multiple linear regressions, whole robustness for the simple linear regression through the origin finds an important application in the estimation of ratios and finite population means. As shown in Figure~\ref{fig_examples}, one may encounter the presence of outliers in achieving this task. In \cite{gwet1992outlier}, the ratio aimed to be estimated was the area under wheat in 1936 to the total cultivated area in 1931 in a given administrative geographical unit of Uttar Pradesh state in India. In \cite{chambers1986outlier}, it was the total population in 1970 in East Baltimore to the number of occupied dwelling in 1960 in the same area. In Section~\ref{sec_sim_study}, we illustrate the relevance of our robust approach through analyses in economic contexts. More precisely, the following contexts are considered: robust estimation of the personal disposable income per capita and of the average weekly household expenditure on food (using the ratio estimator). In Section~\ref{sec_sim_study}, we also detail the link between simple linear regression through the origin and finite population sampling, and present a simulation study. In all analyses, our approach is compared with the nonrobust (with the normal assumption) and partially robust (with the Student distribution assumption) approaches. It is showed that our model performs as well as the nonrobust and the partially robust models in absence of outliers, in addition to being completely robust. It indicates that, by only changing the assumption on the error term, we obtain adequate estimates in absence or presence of outliers. These estimates are computed as usual from the posterior distribution.

    \begin{figure}[h]
  \centering
  $\begin{array}{cc}
  \vspace{-0.0cm}\scriptsize\text{\textbf{Example from \cite{gwet1992outlier}}} & \scriptsize\text{\textbf{Example from \cite{chambers1986outlier}}}\cr
   \hspace{-3mm}\includegraphics[width=6.5cm]{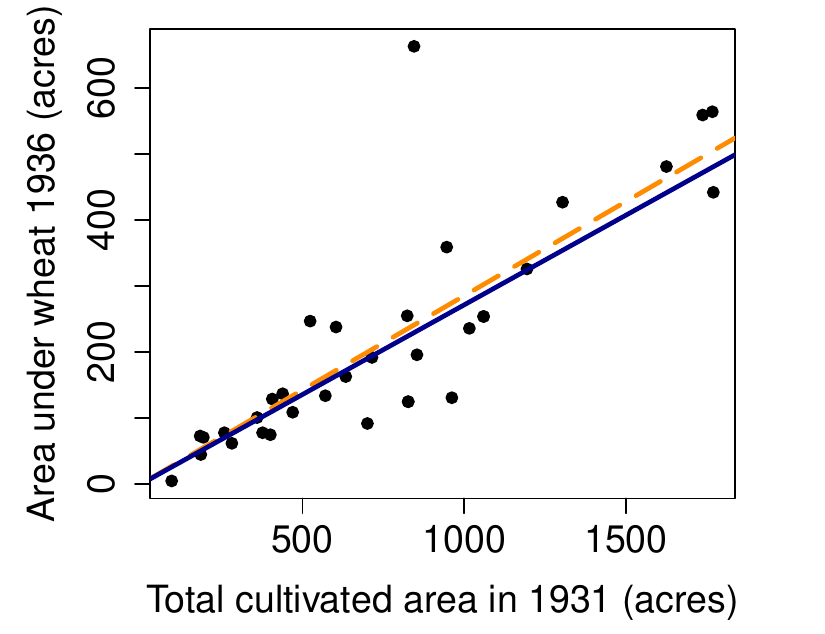} & \hspace{-3mm}\includegraphics[width=6.5cm]{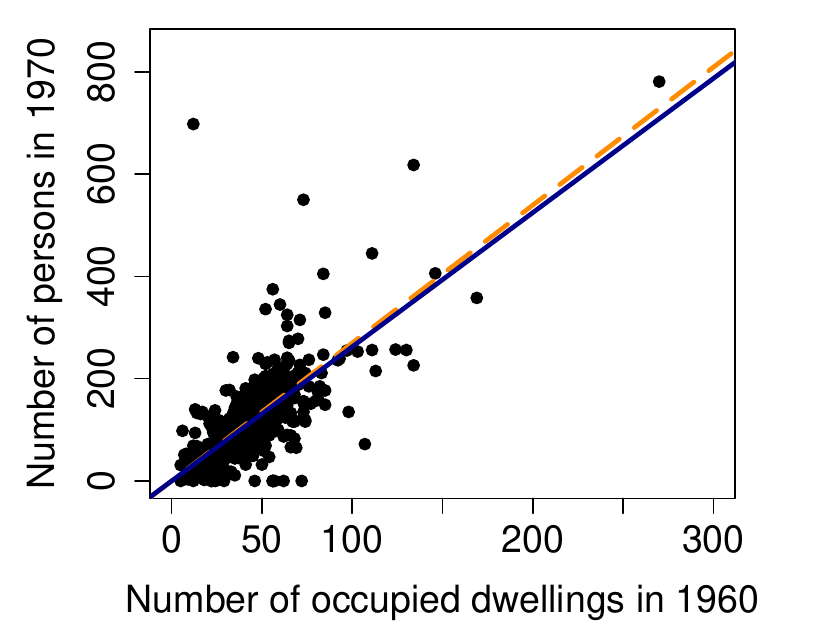}
  \end{array}$
  \caption{\small Example of data sets containing outliers with slope estimates under normal (orange dashed line), and super heavy-tailed (blue solid line) distribution assumptions; the data sets are provided in Section~\ref{sec-supp}}\label{fig_examples}
 \end{figure}
\normalsize

   \section{Resolution of Conflicts in Simple Linear Regression Through the Origin}\label{sec-robustness}

\subsection{Model}\label{sec-model}

\begin{description}
   \item[(i)]Let $Y_1,\ldots,Y_n\in\re$ be $n$ random variables and $x_1,\ldots,x_n\in\re\backslash{\{0\}}$ be $n$ known constants, where $n>2$ is assumed to be known. We assume that
   \begin{equation*}
        Y_i=\beta x_i+\epsilon_i,\quad i=1,\ldots,n,
\end{equation*}
where $\epsilon_1,\ldots,\epsilon_n\in\re$ and $\beta\in\re$ are $n+1$ conditionally independent random variables given $\sigma>0$ with a conditional density for $\epsilon_i$ given by
   \begin{equation*}
    \epsilon_i \mid \beta,\sigma\stackrel{\mathcal{D}}{=}\epsilon_i\mid\sigma \simdist \frac{1}{\sigma |x_i|^{\theta}}f\left(\frac{\epsilon_i}{\sigma |x_i|^{\theta}}\right), \quad i=1,\ldots,n,
  \end{equation*}
  $\theta\in\re$ being a known constant.

  \item[(ii)] We assume that $f$ is a strictly positive continuous probability density function on $\re$ that is symmetric with respect to the origin, and that is such that both tails of $|z| f(z)$ are monotonic, which implies that the tails of $f(z)$ are also monotonic. The density $f$ can have parameters, e.g.\ a shape parameter; however, their value is assumed to be known.

   \item[(iii)] We assume that the prior of $\beta$ and $\sigma$, denoted $\pi(\beta,\sigma)$, is bounded on $\sigma>1$, and is such that $\pi(\beta,\sigma)/(1/\sigma)$ is bounded on $0<\sigma\leq 1$, for all $\beta\in\re$. Together, these assumptions are equivalent to: $\pi(\beta,\sigma)/\max(1,$ $1/\sigma)$ is bounded on $\sigma>0$. A large variety of priors fit within this assumed structure; for instance, this is the case for all proper densities. In addition, non-informative priors such as $\pi(\beta,\sigma)\propto 1/\sigma$, the usual one for this type of random variables, and $\pi(\beta,\sigma)\propto 1$ satisfy these assumptions.
 \end{description}

  From this perspective, $x_1,\ldots,x_n$ represent observations of the explanatory variable, the dependent variable and the error term are respectively represented by the continuous random variables $Y_1,\ldots,Y_n$ and $\epsilon_1,\ldots,\epsilon_n$, and the parameter $\beta$ represents the slope of the regression line. Note that no assumptions are made on the explanatory variable, except that the value 0 cannot be observed.

  The scale of the distribution of the error term is $\sigma |x_i|^{\theta}$ and, therefore, the variability of the errors increases (decreases) as $x_i$ moves away from 0 when $\theta>0$ ($\theta<0$). This model can thus be used in a context of heteroscedasticity. When the classical framework is considered, i.e.\ a frequentist setting with the assumption that $f$ is the standard normal density, $\sigma |x_i|^{\theta}$ also represents the standard deviation of the error $\epsilon_i$. In this situation, the maximum likelihood estimator of $\beta$ is the weighted average of the $y_i/x_i$ given by $\hat{\beta}=\sum_{i=1}^{n}w_i (y_i/x_i)$, where $w_i=\abs{x_i}^{2(1-\theta)}/\sum_{j=1}^{n}\abs{x_j}^{2(1-\theta)}$.

  An important drawback of the classical framework is that outliers have a significant impact on the estimation, due to the normal assumption. In this paper, we study robustness of the estimation of $\beta$ and $\sigma$. The objective is to find sufficient conditions to attain whole robustness.
  The nature of the results presented in Section~\ref{sec-conflict} is asymptotic, in the sense that some $y_i$'s approach $+\infty$ or $-\infty$. The known vector $\mathbf{x_n}:=(x_1,\ldots,x_n)$ is considered as fixed. In Section~\ref{sec_illu_thm}, we explain that studying this theoretical framework is sufficient to attain, in practice, robustness against any type of outliers (i.e.\ outliers because of their extreme $x$ value, extreme $y$ value, or both).

Among the $n$ observations of $Y_1,\ldots,Y_n$, denoted by $\mathbf{y_n}$,
we assume that $k>2$ of them, denoted by the vector $\mathbf{y_k}$,
 form a group of nonoutlying observations, $m$ of them are considered as ``negative slope outliers'', with relatively small (large) values of $y_i$ when $x_i$ is positive (negative),
  and $p$ of them are considered as ``positive slope outliers'', with relatively large (small) values of $y_i$ when $x_i$ is positive (negative),
   with  $k+m+p=n$. Note that we use the letter $m$ for ``minus'' because the related outliers attract the slope towards negative values, and analogously, we use the letter $p$ for ``positive''. For $i=1,\ldots,n$, we define the binary functions $k_i, m_i$ and $p_i$ as follows: if $y_i$ is a nonoutlying value, $k_i=1$; if it is a negative slope outlier, $m_i=1$ and if it is a positive slope outlier, $p_i=1$. These functions take the value of 0 otherwise. Therefore, we have $k_i+m_i+p_i=1$ for $i=1,\ldots,n$, with $\sum_{i=1}^n k_i=k$, $\sum_{i=1}^n m_i=m$ and $\sum_{i=1}^n p_i=p$. We assume that each outlier approaches $-\infty$ or $+\infty$ at its own specific rate, to the extent that the ratio of two outliers is bounded. More precisely, we assume that $y_i=a_i+b_i \omega$, for $i=1,\ldots,n$, where $a_i$ and $b_i$ are constants such that $a_i\in\re$ and
\begin{description}
  \item[(i)] $b_i=0$ if $k_i=1$,
  \item[(ii)] $b_i<0$ if $y_i$ is ``small'', that is if $x_i<0, p_i=1$ or $x_i>0, m_i=1$,
  \item[(iii)] $b_i>0$ if $y_i$ is ``large'', that is if $x_i<0, m_i=1$ or $x_i>0, p_i=1$,
\end{description}
and we let $\omega\rightarrow \infty$.

Let the joint posterior density of $\beta$ and $\sigma$ be denoted by $\pi(\beta,\sigma\mid \mathbf{y_n})$
and the marginal density of $(Y_1,\ldots,Y_n)$ be denoted by $m(\mathbf{y_n})$, where
\begin{equation*}
\pi(\beta,\sigma\mid\mathbf{y_n})=[m(\mathbf{y_n})]^{-1}\pi(\beta,\sigma)\prod_{i=1}^n \frac{1}{\sigma |x_i|^{\theta}}f\left(\frac{y_i-\beta x_i}{\sigma |x_i|^{\theta}}\right),\hspace{1mm}\beta\in\re,\sigma>0.
\end{equation*}
Let the joint posterior density of $\beta$ and $\sigma$ arising from the nonoutlying observations only be denoted by $\pi(\beta,\sigma\mid \mathbf{y_k})$ and the corresponding marginal density be denoted by $m(\mathbf{y_k})$, where
\begin{equation*}\label{eqn-nonoutlier}
\pi(\beta,\sigma\mid\mathbf{y_k})=[m(\mathbf{y_k})]^{-1}\pi(\beta,\sigma)\prod_{i=1}^n \left[\frac{1}{\sigma |x_i|^{\theta}}f\left(\frac{y_i-\beta x_i}{\sigma |x_i|^{\theta}}\right)\right]^{k_i},\hspace{1mm}\beta\in\re,\sigma>0.
\end{equation*}
Note that if the prior $\pi(\beta,\sigma)$ is proportional to 1, the likelihood functions, given by the product term in the posteriors above, can also be expressed as follows:
\begin{equation}\label{likeli1}
\likeli(\beta,\sigma\mid\mathbf{y_n})= m(\mathbf{y_n})\pi(\beta,\sigma\mid\mathbf{y_n})\hspace{2mm}\text{ and }\hspace{2mm}
\likeli(\beta,\sigma\mid\mathbf{y_k})= m(\mathbf{y_k})\pi(\beta,\sigma\mid\mathbf{y_k}).
\end{equation}
\begin{Proposition}\label{proposition-proper}
Considering the Bayesian context given in Section~\ref{sec-model},
the joint posterior densities $\pi(\beta,\sigma\mid\mathbf{y_k})$ and $\pi(\beta,\sigma\mid\mathbf{y_n})$ are proper.
\end{Proposition}
 The proof of Proposition~\ref{proposition-proper} can be found in Section~\ref{sec-supp}.

  \subsection{Log-Regularly Varying Distributions}\label{sec-log-regularly}

 As mentioned in the introduction, our approach to attain robustness is to replace the traditional normal assumption on the error term by a log-regularly varying distribution assumption. The definition of such a distribution is now presented.

 \begin{Definition}[Log-regularly varying distribution]\label{def-log-regularly-distribution}
   A random variable Z with a symmetric density $f(z)$ is said to have a log-regularly varying distribution with index $\rho\geq 1$ if $z f(z)\in L_{\rho}(\infty)$, meaning that $zf(z)$ is \textit{log-regularly varying} at $\infty$ with index $\rho\geq1$.
\end{Definition}

Log-regularly varying functions is an interesting class of functions with useful properties for robustness. By definition, they are such that $g\in L_{\rho}(\infty)$ if $g(z^\nu)/g(z)$ converges towards $\nu^{-\rho}$ uniformly in any set $\nu\in[1/\tau,\tau]$ (for any $\tau\ge 1$) as $z\rightarrow\infty$, where $\rho\in\re$. This implies that for any $\rho\in\re$, we have $g\in L_{\rho}(\infty)$ if and only if there exists a constant $A>1$ and a function $s\in L_{0}(\infty)$ (which is called a log-slowly varying function) such that for $z\ge A$, $g$ can be written as $g(z)=(\log z)^{-\rho} s(z)$. An example of log-regularly varying distributions is presented in Section~\ref{sec_illu_thm}. The purpose of this section was to provide an overview of the tail behaviour of such distributions. For more information on log-regularly varying distributions, we refer the reader to \cite{desgagne2013full} and \cite{desgagne2015robustness}.

\subsection{Resolution of conflicts}\label{sec-conflict}

The results of robustness are now given in Theorem~\ref{thm-main}.
\begin{theorem}\label{thm-main}
Consider the model and the context described in Section~\ref{sec-model}.
If we assume that
 \begin{description}
         \item[(i)]$z f(z)\in L_{\rho}(\infty)$, with $\rho\ge 1$ \quad (i.e.\ that $f$ is a log-regularly varying distribution),
          \item[(ii)]$k>\max(m,p)$ \quad (i.e.\ that both the negative and positive slope outliers are fewer than the nonoutliers),
      \end{description}
then, recalling that $y_i=a_i+b_i \omega$ with $b_i=0$ for the nonoutliers and $b_i\neq 0$ for the outliers, we obtain the following results:
 \begin{description}
   \item[(a)]
   \begin{equation*}
   \lim_{\omega\rightarrow \infty}\frac{m(\mathbf{y_n})}{\prod_{i=1}^{n}[f(y_i)]^{m_i+p_i}}= m(\mathbf{y_k}),
   \end{equation*}
     \item[(b)]
  \begin{equation*}
   \lim_{\omega\rightarrow \infty}\pi(\beta,\sigma\mid\mathbf{y_n})=\pi(\beta,\sigma\mid\mathbf{y_k}),
   \end{equation*}
uniformly on $(\beta,\sigma)\in [-\lambda,\lambda]\times [1/\tau,\tau]$, for any $\lambda\ge 0$ and $\tau\ge 1$,
\item[(c)]
\begin{equation*}
\lim_{\omega\rightarrow \infty}\int_{0}^{\infty}\int_{-\infty}^{\infty}\big|\pi(\beta,\sigma\mid\mathbf{y_n})-\pi(\beta,\sigma\mid\mathbf{y_k})\big|\,d\beta\,d\sigma= 0,
\end{equation*}
\item[(d)]As $\omega\rightarrow \infty$,
\begin{equation*}
 \beta,\sigma\mid\mathbf{y_n} \convdist \beta,\sigma\mid\mathbf{y_k},
 \end{equation*}
and in particular
\begin{equation*}
\beta\mid\mathbf{y_n} \convdist \beta\mid\mathbf{y_k} \hspace{5mm}\text{ and }\hspace{5mm}
 \sigma\mid\mathbf{y_n} \convdist \sigma\mid\mathbf{y_k},
 \end{equation*}
  \item[(e)]
  \begin{equation*}
   \lim_{\omega\rightarrow \infty}[m(\mathbf{y_k})/m(\mathbf{y_n})]\likeli(\beta,\sigma\mid\mathbf{y_n})=\likeli(\beta,\sigma\mid\mathbf{y_k}),
   \end{equation*}
uniformly on $(\beta,\sigma)\in [-\lambda,\lambda]\times [1/\tau,\tau]$, for any $\lambda\ge 0$ and $\tau\ge 1$.
\end{description}
\end{theorem}

The proof of Theorem~\ref{thm-main} can be found in Section~\ref{sec-supp}. Note that, when $x_1=\ldots=x_n=1$, the simple linear regression model through the origin becomes the location-scale model, and this highlights the fact that our results generalise those of \cite{desgagne2015robustness}.

Theorem~\ref{thm-main} is particularly appealing for its simplicity, and therefore, for its practical use. Indeed, condition (i) only indicates that modelling must be done using a density $f$ with sufficiently heavy tails, specifically with a log-regularly varying distribution (see Definition~\ref{def-log-regularly-distribution}). For that purpose, \cite{desgagne2015robustness} introduced the family of log-Pareto-tailed symmetric distributions, which belongs to the family of log-regularly varying distributions and therefore satisfies condition (i).
This new family includes, for instance, piecewise densities constructed from well-known symmetric densities as the normal, uniform or Student by replacing their extremities by log-Pareto tails, i.e.\ tails that behave like $(1/|z|)(\log|z|)^{-\phi}$ with $\phi>1$. A special case of log-Pareto-tailed symmetric distributions, called the log-Pareto-tailed standard normal (LPTN) distribution with parameters $\alpha>1$ and $\phi>1$, is given in Section~\ref{sec_illu_thm}. It exactly matches the standard normal on the interval $[-\alpha,\alpha]$, with log-Pareto tails. This is the super heavy-tailed distribution that we use in our numerical analyses. Note that we can also construct symmetric densities with log-Pareto tails that are not piecewise through transformations of the Pareto distribution. For instance, from a Pareto random variable $Y$ with density $g(y)=\phi\theta^\phi y^{-(\phi+1)}, y>\theta$, we can make the change of variable $|Z|=e^Y-e^\theta \Leftrightarrow Y=\log(|Z|+e^\theta)$ to obtain a double-log-Pareto distribution with density
$$
 f(z)=(1/2)\phi\theta^\phi(|z|+e^\theta)^{-1}[\log(|z|+e^\theta)]^{-(\phi+1)}, \quad  -\infty<z<\infty, \theta>0, \phi>0.
 $$

  Condition (ii) indicates that both the negative and positive slope outliers must be fewer than the  nonoutlying observations, i.e.\ $m<k$ and $p<k$. In other words, the nonoutlying observations must form the largest group. For instance, with a sample of size $n=25$, the model rejects up to 16 outliers if they are split in $m=8$ negative and $p=8$ positive slope outliers, which leaves $k=9$ nonoutliers. At the other end of the spectrum, in the situation where all outliers are of the same type, for instance all positive slope outliers (which implies that $m=0$), the model rejects up to $p=12$ outliers, which leaves $k=13$ nonoutliers. Numerical simulations seem to confirm our expectation that a larger difference between $k$ and $\max(m,p)$ results in a more rapid rejection of the outliers.

  The breakdown point is generally defined as the largest proportion of outliers that an estimator can handle. In our situation, for a sample size of $n$, the condition $k>\max(m,p)$ translates into a breakdown point of $\lfloor(n-1)/2 \rfloor/n$, that is the integer part of $(n-1)/2$ divided by $n$, if we consider only positive slope outliers (or only negative slope outliers). As $n\rightarrow\infty$, the breakdown point converges to $0.5$, usually considered as the maximum desired value.

Not only do the conditions of Theorem~\ref{thm-main} are simple and intuitive, the results are also easy to interpret. The asymptotic behaviour of the marginal $m(\mathbf{y_n})$ is described by result~(a). While this result is more of theoretical interest, it is the cornerstone of this robustness theory; it leads to results (b) to (e), which are more practical. Result~(b) indicates that the posterior density, arising from the whole sample, converges towards the posterior density arising from the nonoutliers only, uniformly in any set $(\beta,\sigma)\in [-\lambda,\lambda]\times [1/\tau,\tau]$. The impact of the outliers then gradually decreases to nothing as they approach plus or minus infinity.

Result~(b) leads to result~(c): the convergence in $L_1$ of the posterior density, arising from the whole sample, towards the posterior density arising from the nonoutlying observations only. This last result implies the following convergence: $\Prob(\beta,\sigma\in E \mid\mathbf{y_n})\rightarrow\Prob(\beta,\sigma\in E \mid\mathbf{y_k})$ as $\omega\rightarrow\infty$, uniformly for all rectangles $E\in\re\times \re^{+}$. This result is slightly stronger than convergence in distribution (result~(d)) which requires only pointwise convergence. Then, the convergence of the posterior marginal distributions is directly obtained. Therefore, any estimation of $\beta$ and $\sigma$ based on posterior quantiles (e.g.\ posterior medians and Bayesian credible intervals) is robust to outliers.
Note that results~(a) to (d) are also valid if we assume that $n\ge 2, k\ge 2$ (instead of $n>2, k>2$), provided that we assume that
 $\sigma\pi(\beta,\sigma)$ is bounded (instead of $\min(\sigma,1)\pi(\beta,\sigma)$ is bounded).

Result~(e) indicates that, for a given sample, the likelihood (up to a multiplicative constant   that does not depend on  $\beta$ and $\sigma$) converges to the likelihood arising from the nonoutliers only, uniformly in any set $(\beta,\sigma)\in E$,   where $E=[-\lambda,\lambda]\times [1/\tau,\tau]$. Consequently, the maximum of $\likeli(\beta,\sigma\mid\mathbf{y_n})$ thus converges to the maximum of $\likeli(\beta,\sigma\mid\mathbf{y_k})$ on the set $E$ and, therefore the maximum likelihood estimate also converges, as $\omega\rightarrow\infty$.
Note that, using results (b) to (d), we know that, for both $\pi(\beta,\sigma\mid\mathbf{y_k})$  and $\pi(\beta,\sigma\mid\mathbf{y_n})$, the volume on $E^c$ over
    the volume on $E$ converges to 0 as $\lambda$ and $\tau$ increase; this relation holds in particular if $\pi(\beta,\sigma)\propto 1$
    and, in this case, the posterior is proportional to the likelihood.

 \section{Finite Population Means and Ratios}\label{sec_sim_study}

 To use the model described in Section~\ref{sec-model}, users have to set the value of $\theta$. Different particular values lead to interesting special cases. For instance, when $\theta=0$, the resulting model is the classical homoscedastic model, with $\text{Var}(\epsilon_i)=\sigma^2$ and $\hat{\beta}=\sum_{i=1}^{n}x_i y_i\big/ \sum_{j=1}^n x_j^2$, considering the classical framework. When $\theta=1$, the estimator of $\beta$ is the unweighted mean of the $y_i/x_i$, that is $\hat{\beta}=(1/n)\sum_{i=1}^{n}y_i/x_i$. Probably the most interesting special case results from $\theta=1/2$ and $x_i>0$ for all $i$. Indeed, considering again the classical framework, the estimator of $\beta$ is $\hat{\beta}=\sum_{i=1}^{n}y_i\big/\sum_{i=1}^{n}x_i$, which is commonly used to estimate the following finite population ratio: $\sum_{i=1}^{N}y_i\big/\sum_{i=1}^{N}x_i$, where $y_i$ and $x_i$ are measures of the variable of interest and of the auxiliary variable on unit $i$, respectively, and $N$ is the population size. The estimator $\hat{\beta}=\sum_{i=1}^n y_i\big/\sum_{i=1}^n x_i$ is also used to estimate the finite population mean $\mu_y$ of a variable of interest $y$ using auxiliary information of a variable $x$ as follows: $\hat{\mu}_y= \hat{\beta}\times\mu_x$, where $\mu_x$ is the known population mean of $x$. This last estimator is known as the ratio estimator and to be more accurate than the simple location model when the variable of interest is correlated with the auxiliary variable. Therefore, robust estimators of $\beta$ lead to robust estimators of finite population means and ratios.  To our knowledge, \cite{gwet1992outlier} introduced the first frequentist outlier resistant alternatives to the ratio estimator, using well known \textit{M}- (\cite{huber1973robust}) and \textit{GM}- (\cite{mallows1975some}) estimators. Their research was inspired by the work of \cite{chambers1986outlier}, the first author to use regression \textit{M}-estimators in survey sampling.

 In Section~\ref{sec_illu_thm}, we present real-life situations in which ratio estimation is useful, while illustrating the theoretical results of Theorem~\ref{thm-main}. First, in a context of estimation of personal disposable income (PDI) per capita, we show that, when we artificially move an observation, its impact on the estimation grows until it reaches a certain threshold.
 Beyond this threshold, the impact decreases to nothing as the observation approaches plus or minus infinity.
 Second, a more traditional Bayesian analysis is made, in which we study the proportion of income spent on food. More precisely, we present the posterior distributions, with particular emphasis on the impact of outliers, and we compute various estimates from the posteriors. In Section~\ref{sec_comparison}, again in a context of finite population sampling, a simulation study is conducted to evaluate the accuracy of the estimates arising from our model. In all analyses, we compare its performance with those of the nonrobust (the model with the normal assumption) and partially robust (the model with the Student distribution assumption) models. As mentioned in  Section~\ref{intro}, the model of \cite{1968tiao119} can be viewed as a special case of the partially robust model. We therefore omit the comparison with their model. In the simulation study, we also consider the following frequentist competitors: the \textit{M}- and \textit{S}- (\cite{rousseeuw1984robust}) estimators. R functions that are used for the computations are provided in Section~\ref{sec-supp}.

\subsection{Illustration of the Results of Theorem~\ref{thm-main}}\label{sec_illu_thm}

In the first context, we are interested in the estimation of the PDI per capita when the available data are the total disposable income ($y_i$) for $n$ households (in this analysis $n=20$), and the number of individuals ($x_i$) in each of these households. The data are presented in Table~\ref{tab_data}. The PDI per capita, which is a population mean per individual, would be directly computed by $\sum_{i=1}^N y_i\big/\sum_{i=1}^N x_i$ (where $N$ is the number of households in the population) if the information was available for all the households. We therefore use the simple linear regression model through the origin with $\theta=1/2$ to estimate this ratio (see Section~\ref{sec-model} for details about the model).

 \begin{table} [h]
  \centering
  \footnotesize
   \begin{tabular}{rrrrrrrrrrrrrrrrrrrrr}
    \toprule
     $y_i$ & 20.8 & 9.6 & 38.6 & 74.1 & 108.8 & 98.7 & 44.8 & 77.2 & 93.2 & 107.2 \cr
     $x_i$ & 1.0 & 1.0 & 2.0 & 3.0 & 3.0 & 3.0 & 3.0 & 3.0 & 3.0 & 3.0  \\
    \midrule
     $y_i$  & $y_{11}$ & 93.6 & 113.7 & 123.5 & 93.5 & 148.1 & 147.1 & 154.0 & 149.5 & 173.5 \\
     $x_i$ & 3.0 & 4.0 & 4.0 & 4.0 & 4.0 & 5.0 & 5.0 & 5.0 & 6.0 & 6.0 \cr
    \bottomrule
 \end{tabular}
 \caption{\small Total disposable income for household $i$ in thousands of dollars ($y_i$) and the number of individuals in household $i$ ($x_i$), for $i=1,\ldots,20$} \label{tab_data}
 \end{table}

In order to illustrate the threshold feature, an observation is randomly chosen (in this analysis, it is the 11th observation), and $y_{11}$ is gradually moved
from the value 100 (a nonoutlier) to 385 (a large outlier), while $x_{11}=3$ remains fixed. The parameters $\beta$ and $\sigma$ are estimated for each data set related to a different value of $y_{11}$ using maximum \textit{a posteriori} probability (MAP) estimation with a prior proportional to 1 (which corresponds to maximum likelihood estimation). This process is performed under three models, each corresponding to a different assumption on $f$: a standard normal density (in this case, $\hat{\beta}=\sum_{i=1}^{20} y_i\big/\sum_{i=1}^{20} x_i$,
 the classical ratio estimator), a Student density (the partially robust model) or a LPTN density (our robust model). The results are presented in Figure~\ref{fig_illus_thm}.

 \begin{figure}[h]
  \centering
  $\begin{array}{cc}
  \vspace{-0.1cm}\scriptsize\text{\textbf{Estimation of the PDI per capita ($\beta$)}} & \scriptsize\text{\textbf{Estimation of $\sigma$ under various dist.}}\cr
   \vspace{-0.1cm}\scriptsize\text{\textbf{under various dist. assumptions when}} & \scriptsize\text{\textbf{assumptions when $y_{11}$ increases}} \cr
   \scriptsize\text{\textbf{$y_{11}$ increases from 100 to 385}} & \scriptsize\text{\textbf{from 100 to 385}} \cr
   \hspace{-3mm}\includegraphics[width=6.5cm]{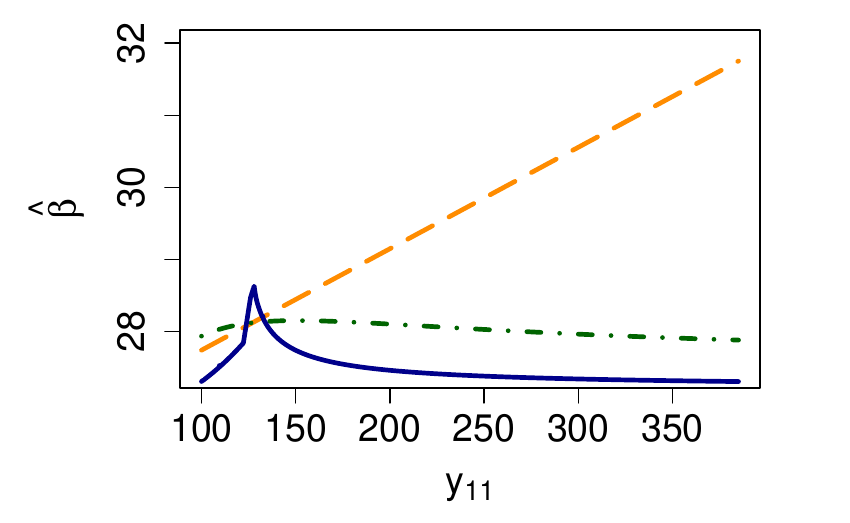} & \hspace{-3mm}\includegraphics[width=6.5cm]{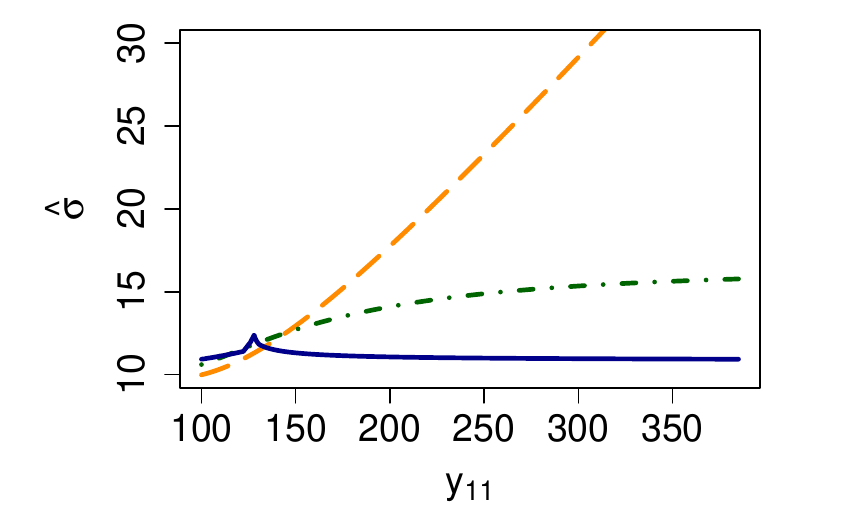}
  \end{array}$
  \caption{\small Estimation of the PDI per capita ($\beta$) and $\sigma$ when $y_{11}$ increases from 100 to 385 under three different assumptions on $f$: standard normal density (orange dashed line), Student density (green dot-dashed line) and LPTN density (blue solid line)}\label{fig_illus_thm}
 \end{figure}
\normalsize

The inference is clearly not robust when it is assumed that the error has a normal distribution (orange dashed line) since the values of the point estimates of $\beta$ and $\sigma$ increase with $y_{11}$. Regarding the second model, the degrees of freedom of the heavy-tailed Student distribution have been arbitrarily set to 10 and a known scale parameter of $0.88$ has been added to this distribution in order to have the same 2.5th and 97.5th percentiles as the standard normal. The estimation of $\beta$ is robust as the impact of the outlier slowly decreases after a certain threshold. However, the estimation of $\sigma$ is only partially robust, i.e.\ the impact of the outlier is limited, but does not decrease when the outlying value increases. For the last model, we set $\alpha$ of the LPTN to $1.96$ so that this distribution matches the standard normal on the interval $[-1.96,1.96]$, implying that both distributions have the same 2.5th and 97.5th percentiles. Therefore, all three distributions studied in this section have 95\% of their mass in the interval $[-1.96,1.96]$. The other parameter of the LPTN $\phi$ is equal to $4.08$ according to the procedure described in Section~4 of \cite{desgagne2015robustness} (this procedure ensures that $f$ is continuous and a probability density function). The density of the LPTN distribution, 
depicted in Figure~\ref{fig_comp_normal_log_pareto}, is given by
\begin{equation}\label{eqn_log_pareto}
\small  f(x)=\begin{cases}
        \varphi(x) \text{ if } \abs{x}\leq \alpha\,\,\text{(the standard normal part)}, \cr
	\varphi(\alpha)(\alpha/\abs{x})(\log \alpha/\log\abs{x})^\phi \text{ if } \abs{x}>\alpha\,\,\text{(the log-Pareto tails)},
       \end{cases}
 \end{equation}
where $\phi=1+2 \varphi(\alpha)\alpha \log(\alpha)/(1-q)$ and $q=\Phi(\alpha)-\Phi(-\alpha)$, $\varphi$ and $\Phi$ being the probability density function and cumulative distribution function of the standard normal distribution, respectively.


 \begin{figure}[h]
  \centering
  $\begin{array}{cc}
    \multicolumn{2}{c} {\vspace{2mm}\small\textbf{ Standard Normal Vs LPTN}} \cr
\includegraphics[width=6cm]{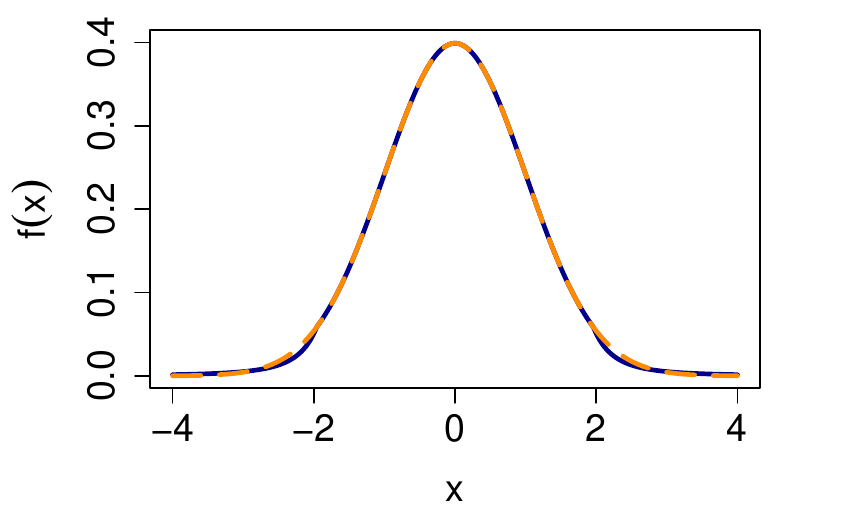} & \includegraphics[width=6cm]{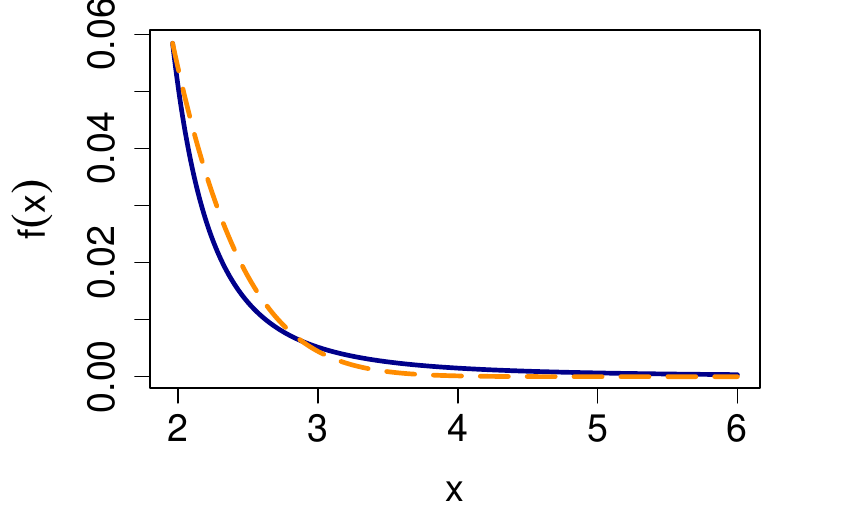}
  \end{array}$
  \caption{\small Densities of the standard normal (orange dashed line) and of the LPTN with $\alpha=1.96$ and $\phi=4.08$ (blue solid line)}\label{fig_comp_normal_log_pareto}
 \end{figure}
 \normalsize

For our robust model, it can be seen that $y_{11}$ has an increasing impact on the estimation until this observation reaches a threshold. In this analysis, the threshold is around $y_{11}=127.9$, and based on the data set with $y_{11}=127.9$, $\hat{\beta}=28.6$ and $\hat{\sigma}=12.4$, which is interpreted as: the personal disposable income per capita is approximately 28,600. Beyond this threshold, the impact of the outlier gradually decreases to nothing as the conflict grows infinitely. The point estimates converge towards $27.1$ for $\beta$ and $10.8$ for $\sigma$, which are the point estimates when $(x_{11},y_{11})$ is excluded from the sample. Whole robustness is therefore attained for both $\beta$ and $\sigma$. Note that an increase in the value of the parameter $\alpha$ would result in an increase in the value of the threshold. Setting $\alpha=1.96$ seems to be suitable for practical use.

In the second context, we are interested in the estimation of the proportion of weekly income spent on food for a population, when data are available per household. If the information was available for the population, we would directly compute the proportion by $\sum_{i=1}^N y_i\big/\sum_{i=1}^N x_i$, where $N$ is the number of households in the population, and $y_i$ and $x_i$ are respectively the weekly expenditure on food and the weekly income, for household $i$. This ratio can thus be approximated using the simple linear regression through the origin with $\theta=1/2$, and again, we compare our robust model with the nonrobust and partially robust models. We use the same Student and LPTN distributions as in the first context above, but we set the prior $\pi(\beta,\sigma)\propto1/\sigma$. A Markov chain Monte Carlo (MCMC) method is implemented for the estimation (see Section~\ref{sec-supp} for the R functions). It is run for 10,000,000 iterations.

Note that the ratio $\sum_{i=1}^N y_i\big/\sum_{i=1}^N x_i$ can be viewed as the following weighted average: $\sum_{i=1}^N w_i (y_i/$ $x_i)$, where $w_i:=x_i/\sum_{i=1}^N x_i$. It means that proportion of weekly income spent on food for a population, $\sum_{i=1}^N y_i\big/\sum_{i=1}^N x_i$, is also a weighted average of proportions of weekly income spent on food per household, where the weight is proportional to the weekly income.

The data set, comprised of the weekly expenditures on food and weekly incomes for twenty households, is presented in Table~\ref{tab_data_food}  and depicted in Figure~\ref{fig_food_with_outliers} (a). The posterior distributions of $\beta$ and $\sigma$ are presented in Figures~\ref{fig_food_with_outliers} (b) and (c). The posterior medians of $\beta$ are $0.283$, $0.306$ and $0.319$ with 95\% highest posterior density (HPD) intervals of $(0.217,0.348)$, $(0.243,0.367)$ and $(0.240,0.376)$ for the nonrobust, partially robust and robust models, respectively. As a result, the proportion of weekly income spent on food for this population is estimated at $0.319$ (considering our robust model) with a 95\% HPD interval of $(0.240,0.376)$. The average weekly household expenditure on food of this population can also be estimated using the ratio estimator. Considering our robust model, it is estimated at $\hat{\mu}_y=\hat{\beta}\times \mu_x=0.319\times 210=66.99$ (considering an average weekly household income of 210 for this population) with a 95\% HPD interval of $(50.40, 78.96)$. The posterior medians of $\sigma$ are $2.180$, $2.031$ and $1.634$ with 95\% HPD intervals of $(1.565,3.016)$, $(1.319,2.958)$ and $(0.962,2.674)$, for the nonrobust, partially robust and robust models, respectively.

\begin{table} [h]
  \centering
  \footnotesize
   \begin{tabular}{rrrrrrrrrrrrrrrrrrrrr}
    \toprule
      $y_i$ & 31.7 & 68.4 & 54.4 & 53.5& 78.4 &66.4 & 64.1 & 44.6 & 99.0 & 53.3 & \cr
     $x_i$ & 102.9 & 144.9 & 155.8 & 176.5 & 177.4 & 182.2 & 197.9 & 199.2 & 211.3 & 215.9   \cr
      \midrule
     $y_i$ & 67.3 & 68.6 & 63.0 & 100.6 & 82.2 & 113.4 & 6.1 & 76.6 & 92.7 & 41.1 \cr
     $x_i$ & 216.0 & 216.7 & 220.3 & 222.8 & 229.0 & 250.0 & 250.2 & 275.4 & 342.4 & 696.4 \cr
    \bottomrule
 \end{tabular}
 \caption{\small Weekly expenditure on food ($y_i$) and weekly income ($x_i$) for household $i$ in dollars, $i=1,\ldots,20$} \label{tab_data_food}
 \end{table}

   We observe the presence of two clear outliers: $(x_{17},y_{17})=(250.2,6.1)$ (because of its extremely low $y$ value) and $(x_{20},y_{20})=(696.4,41.1)$ (because of its extremely high $x$ value). In order to draw conclusions based on the bulk of the data and to evaluate the impact of outliers, we redo the analysis while excluding these two outliers. The results are presented in Figure~\ref{fig_food_without_outliers}. The estimates arising from the three models are now similar. The posterior medians of $\beta$ are $0.342$, $0.339$ and $0.343$ with 95\% HPD intervals of $(0.302,0.382)$, $(0.298,0.380)$ and $(0.303,0.382)$ for the nonrobust, partially robust and robust models, respectively. Therefore, the proportion of weekly income spent on food for this population is estimated at $0.343$ (considering our robust model) with a 95\% HPD interval of $(0.303,0.382)$, based on the bulk of the data. Considering our robust model, the average weekly household expenditure on food is now estimated at $\hat{\mu}_y=\hat{\beta}\times \mu_x=0.343\times 210=72.03$ (considering an average weekly household income of 210 for this population) with a 95\% HPD interval of $(63.63, 80.22)$, using the ratio estimator. The posterior medians of $\sigma$ are $1.177$, $1.268$ and $1.190$ with 95\% HPD intervals of $(0.825,1.656)$, $(0.850, 1.823)$ and $(0.854,1.661)$, for the nonrobust, partially robust and robust models, respectively.

  \begin{figure}[h]
    $\begin{array}{ccc}
   \hskip-0.2cm\scriptsize\text{\textbf{(a) Income spent on food}} & \hspace{-2mm} \scriptsize\text{\textbf{(b) Posterior density of $\beta$}} &\hspace{-1mm} \scriptsize\text{\textbf{(c) Posterior density of $\sigma$}} \cr
   \hskip-0.2cm\includegraphics[width=4.1cm]{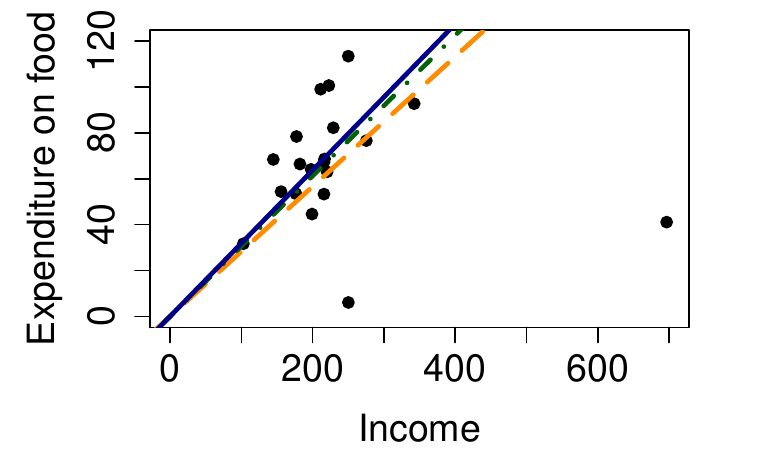} &\hspace{-2mm} \includegraphics[width=4.1cm]{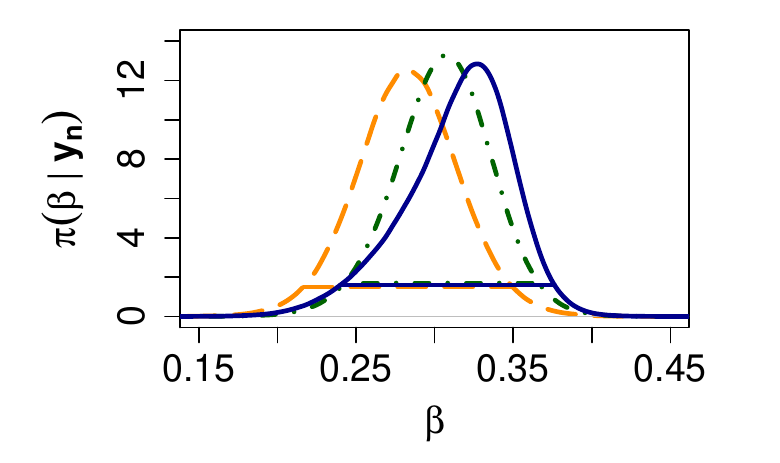} &\hspace{-1mm} \includegraphics[width=4.1cm]{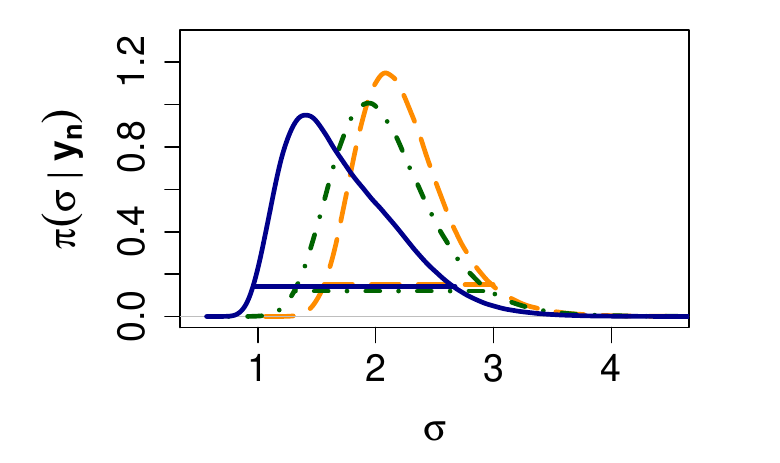}
  \end{array}$
  \caption{\small Expenditure on food as a function of the income with an estimation of the expenditure on food $\hat{\beta} x_i$ based on the posterior median, (b)-(c) Posterior densities of $\beta$ and $\sigma$ arising from the original data with 95\% HPD intervals (horizontal lines); for each graph, the orange dashed, green dot-dashed and blue solid lines are respectively related to the nonrobust, partially robust and robust models}\label{fig_food_with_outliers}
 \end{figure}
\normalsize

  \begin{figure}[h]
  $\begin{array}{ccc}
   \hskip-0.2cm\scriptsize\text{\textbf{(a) Income spent on food}} & \hspace{-2mm} \scriptsize\text{\textbf{(b) Posterior density of $\beta$}} & \hspace{-1mm} \scriptsize\text{\textbf{(c) Posterior density of $\sigma$}} \cr
   \hskip-0.2cm\includegraphics[width=4.1cm]{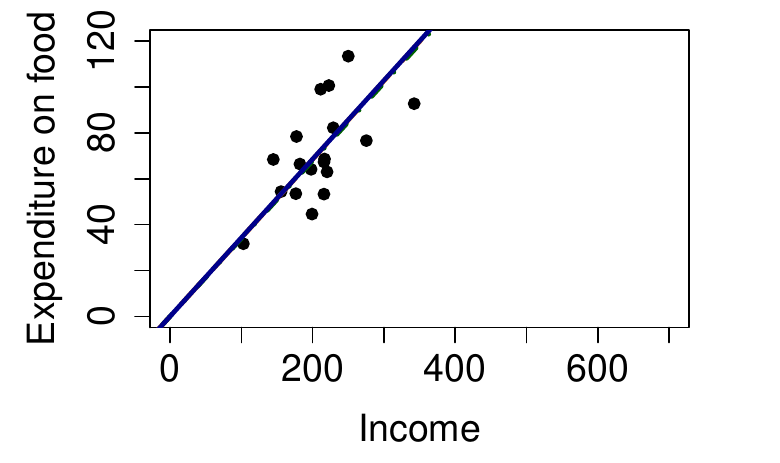} &\hspace{-2mm} \includegraphics[width=4.1cm]{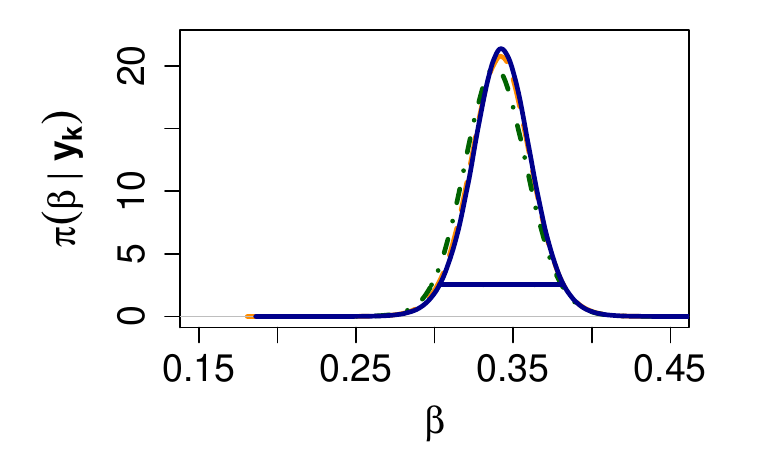} & \hspace{-1mm} \includegraphics[width=4.1cm]{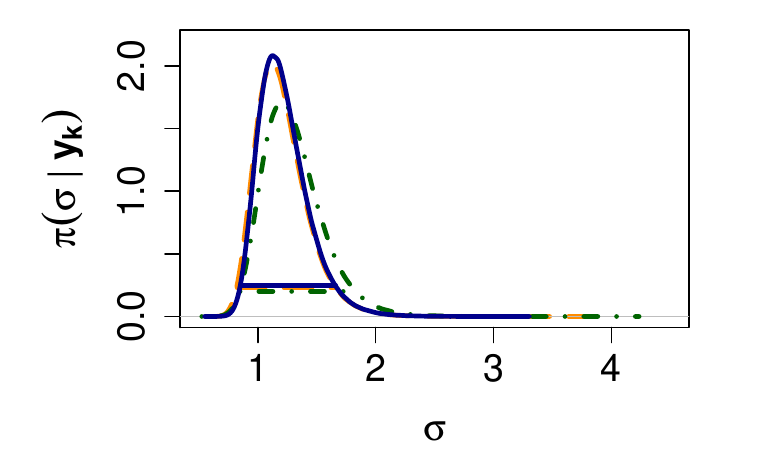}
  \end{array}$
  \caption{\small Expenditure on food as a function of the income with an estimation of the expenditure on food $\hat{\beta} x_i$ based on the posterior median, when the outliers are excluded, (b)-(c) Posterior densities of $\beta$ and $\sigma$ arising from the data set excluding the outliers with 95\% HPD intervals (horizontal lines); for each graph, the orange dashed, green dot-dashed and blue solid lines are respectively related to the nonrobust, partially robust and robust models}\label{fig_food_without_outliers}
 \end{figure}
\normalsize

 Based on the original data set, the inference arising from our robust model is the one that best reflects the behaviour of the bulk of the data, compared to the inferences arising from the non robust and partially robust models. Our robust model therefore succeeds in limiting the influence of outliers in order to obtain conclusions consistent with the majority of the observations.

 Note that an outlier with an extreme $x$ value, as $(x_{20},y_{20})=(696.4,41.1)$, can be viewed as an observation with a fixed $x$ value and an extreme $y$ value (in this case, as an observation with a fixed $x$ value of $696.4$ and an extremely low $y$ value of $41.1$, compared to the trend emerging from the bulk of the data). This explains why our robust model produces robust inference in the presence of this type of outliers.

\subsection{Simulation Study}\label{sec_comparison}

We now evaluate the accuracy of the estimates arising from our robust model in a context of finite population sampling. More precisely, the model $Y_i= \beta x_i+\epsilon_i$ with $\epsilon_i\mid\sigma \simdist 1/(\sigma x_i^{1/2})f(\epsilon_i/(\sigma x_i^{1/2}))$ and $x_i>0$, $i=1,\ldots,n$, is used to analyse the data, where $f$ is assumed to be a LPTN density in our robust model. We consider two sets of parameters for the LPTN: $(\alpha,\phi)=(1.96, 4.08)$ as in Section~\ref{sec_illu_thm}, and $(\alpha,\phi)=(1.50, 2.18)$. Our model is compared with the same linear regression model, but where $f$ is assumed to be a standard normal density in the nonrobust model, and where $f$ is assumed to be a Student density with 10 degrees of freedom and a known scale parameter of $0.88$ in the partially robust model, as in Section~\ref{sec_illu_thm}. We set $\pi(\beta,\sigma)\propto 1$ and we estimate $\beta$ and $\sigma$ using MAP estimation (which therefore corresponds to maximum likelihood estimation) for these three models. Given that the obtained estimates are the same as under the frequentist paradigm, we also compare with the \textit{M}- and \textit{S}-estimators.

We set $n=20$ and $x_1,x_2,\ldots,x_{20}=1,2,\ldots,20$. We simulate 1,000,000 data sets using values for $\beta$ and $\sigma$ arbitrarily set to 1 and 1.5, respectively, and we carry out this process for each of the three scenarios that we now describe.
In the first one, $f$ is a standard normal distribution; therefore, the probability to observe outliers is negligible. In the second scenario, $f$ is a mixture of two normals where the first component is a standard normal distribution and the second has a mean of 0 and a variance of $10^2$, with weights of $0.9$ and $0.1$, respectively.
 This last component can contaminate the data set by generating extreme values. In the third and last scenario, $f$ is also a mixture of two normals, but the contamination is due to the second component's location. More precisely, the first component is again a standard normal, but the second has a mean of 10 and a variance of 1, with weights of $0.95$ and $0.05$, respectively.

Within each simulation scenario, we evaluate the performance of each model and estimator using sample mean square errors (MSE), based on the true values $\beta=1$ and $\sigma=1.5$. The results are presented in Tables~\ref{tab_comparison_beta} and \ref{tab_comparison_sigma}.

If we first compare the models that we considered in Section~\ref{sec_illu_thm} (the models with the normal, Student and LPTN with $\alpha=1.96$ and $\phi=4.08$ assumptions), we observe that they have \textit{almost} identical performances for both the estimation of $\beta$ and $\sigma$, when there are no outliers (the 100\% $\mathcal{N}(0,1)$ scenario). This was expected given that the three related densities are very similar, especially on the interval $[-1.96,1.96]$ where they all have 95\% of their mass. They however differ in the thickness of their tails, and this feature plays a major role when the sample contains outliers, which is frequently the case for the two other scenarios. As expected, the presence of outlying observations has a major impact on the estimations when the traditional standard normal assumption is used. For the model with the Student distribution assumption, outliers influence the estimation of $\sigma$ significantly, while having a lesser effect on $\hat{\beta}$, which reflects the partial robustness of this approach. The impact on the estimation of both $\beta$ and $\sigma$ is limited for our robust alternative, as suggested by the theoretical results.

Our robust model with $\alpha=1.96$ and $\phi=4.08$ performs better than the frequentist competitors regarding the estimation of $\sigma$ in the absence of outliers. The latter however produce more accurate estimates in the probable presence of outliers. There is a trade-off between the extent to which a model (or a loss function for the frequentist competitors) matches the traditional normal one, and the level of robustness it features. We clearly observe this by decreasing the value for $\alpha$ of the LPTN to 1.5, which leads to a density that matches that of the normal on $[-1.5,1.5]$ (instead of on $[-1.96,1.96]$), but has heavier tails ($\phi=2.18$).


 \begin{table} [H]
 \centering
 \tiny
  \begin{tabular}{lrrr}
   \toprule
      \textbf{Assumptions on \textit{f}} & \multicolumn{3}{c}{\textbf{Scenarios}} \cr
      \cmidrule{2-4}
           & $100\% \mathcal{N}(0,1)$ & $90\% \mathcal{N}(0,1)+10\% \mathcal{N}(0,10^2)$ & $95\% \mathcal{N}(0,1)+5\% \mathcal{N}(10,1)$ \cr
   \midrule
     Standard normal & 0.011 & 0.117 & 0.110 \cr  
     Student (10 d.f.) & 0.011 & 0.027 & 0.033 \cr 
     \midrule
     LPTN \cr
     with $\alpha=1.96$ and $\phi=4.08$ & 0.011 &  0.020 &  0.018  \cr 
     with $\alpha=1.50$ and $\phi=2.18$ & 0.013 &  0.016 &  0.013  \cr 
     \midrule
     \textit{M}-estimator & 0.011 & 0.017 & 0.017\cr 
     \textit{S}-estimator & 0.029 & 0.027 & 0.027 \cr 
   \bottomrule
\end{tabular}
\caption{\small MSE of the estimators of $\beta$ under the three scenarios and the three assumptions of $f$} \label{tab_comparison_beta}
\end{table}

\begin{table} [H]
 \centering
 \tiny
  \begin{tabular}{lrrr}
   \toprule
        \textbf{Assumptions on \textit{f}} & \multicolumn{3}{c}{\textbf{Scenarios}} \cr
      \cmidrule{2-4}
           & $100\% \mathcal{N}(0,1)$ & $90\% \mathcal{N}(0,1)+10\% \mathcal{N}(0,10^2)$ & $95\% \mathcal{N}(0,1)+5\% \mathcal{N}(10,1)$ \cr
   \midrule
     Standard normal & 0.06 & 12.98 & 5.03\cr 
     Student (10 d.f.) & 0.06 & 4.02 & 1.99\cr 
     \midrule
     LPTN \cr
     with $\alpha=1.96$ and $\phi=4.08$ &  0.07 &  0.60 &  0.22   \cr 
     with $\alpha=1.50$ and $\phi=2.18$ &  0.09 &  0.20 &  0.11  \cr 
     \midrule
     \textit{M}-estimator & 0.14 & 0.24 & 0.21\cr 
     \textit{S}-estimator & 0.11 & 0.23 & 0.15 \cr 
   \bottomrule
\end{tabular}
\caption{\small MSE of the estimators of $\sigma$ under the three scenarios and the three assumptions of $f$} \label{tab_comparison_sigma}
\end{table}


\section{Conclusion}\label{sec_conclusion}

In this paper, we have provided a simple Bayesian approach to robustly estimate both parameters $\beta$ and $\sigma$ of a simple linear regression through the origin, in which the variance of the error term can depend on the explanatory variable. It leads to robust estimators of finite population means and ratios. The approach is to replace the traditional normal assumption on the error term by a super heavy-tailed distribution assumption. In particular, we considered log-regularly varying distributions. Whole robustness is attained provided that both the negative and positive slope outliers are fewer than the nonoutlying observations, i.e.\ $m<k$ and $p<k$, as stated in Theorem~\ref{thm-main}.

The theoretical results have been illustrated in Section~\ref{sec_sim_study} through typical real-life situations in which ratio estimation is used, and a simulation study. All the analyses leading to robust inference have been done using the log-Pareto-tailed standard normal (LPTN) density given in~(\ref{eqn_log_pareto}). Our model has been compared with the nonrobust (with the normal assumption) and partially robust (with the Student distribution assumption) models. The conclusion is: our model performs as well as the nonrobust and the partially robust models in absence of outliers, in addition to being completely robust. Therefore, our recommendation is to assume that the error has the density given in (\ref{eqn_log_pareto}) and obtain adequate results, regardless of whether there are outliers, by computing estimates as usual from the posterior distribution.

\section{Acknowledgements}

The authors acknowledge support from the NSERC (Natural Sciences and Engineering  Research  Council  of  Canada),  the  FRQNT  (Le Fonds de recher\-che du Qu\'{e}bec - Nature et technologies) and the SOA (Society of Actuaries). They also would like to thank the anonymous referees for their very helpful comments.

   \bibliographystyle{imsart-nameyear}
\bibliography{reference}

\begin{thebibliography}{16}

\bibitem[\protect\citeauthoryear{Andrade and
  O'Hagan}{2011}]{andrade2011bayesian}
\begin{barticle}[author]
\bauthor{\bsnm{Andrade},~\bfnm{Jose Ailton~Alencar}\binits{J.~A.~A.}} \AND
  \bauthor{\bsnm{O'Hagan},~\bfnm{Anthony}\binits{A.}}
(\byear{2011}).
\btitle{Bayesian Robustness Modelling of Location and Scale Parameters}.
\bjournal{Scand. J. Stat.}
\bvolume{38}
\bpages{691--711}.
\end{barticle}
\endbibitem

\bibitem[\protect\citeauthoryear{Box and Tiao}{1968}]{1968tiao119}
\begin{barticle}[author]
\bauthor{\bsnm{Box},~\bfnm{G.~E.~P.}\binits{G.~E.~P.}} \AND
  \bauthor{\bsnm{Tiao},~\bfnm{G.~C.}\binits{G.~C.}}
(\byear{1968}).
\btitle{A {B}ayesian Approach to Some Outlier Problems}.
\bjournal{Biometrika}
\bvolume{55}
\bpages{119-129}.
\end{barticle}
\endbibitem

\bibitem[\protect\citeauthoryear{Chambers}{1986}]{chambers1986outlier}
\begin{barticle}[author]
\bauthor{\bsnm{Chambers},~\bfnm{Raymond~L}\binits{R.~L.}}
(\byear{1986}).
\btitle{Outlier Robust Finite Population Estimation}.
\bjournal{J. Amer. Statist. Assoc.}
\bvolume{81}
\bpages{1063--1069}.
\end{barticle}
\endbibitem

\bibitem[\protect\citeauthoryear{Desgagn{\'e}}{2013}]{desgagne2013full}
\begin{barticle}[author]
\bauthor{\bsnm{Desgagn{\'e}},~\bfnm{Alain}\binits{A.}}
(\byear{2013}).
\btitle{Full Robustness in {B}ayesian Modelling of a Scale Parameter}.
\bjournal{Bayesian Anal.}
\bvolume{8}
\bpages{187--220}.
\end{barticle}
\endbibitem

\bibitem[\protect\citeauthoryear{Desgagn{\'e}}{2015}]{desgagne2015robustness}
\begin{barticle}[author]
\bauthor{\bsnm{Desgagn{\'e}},~\bfnm{Alain}\binits{A.}}
(\byear{2015}).
\btitle{Robustness to Outliers in Location--Scale Parameter Model using
  Log-Regularly Varying Distributions}.
\bjournal{Ann. Statist.}
\bvolume{43}
\bpages{1568--1595}.
\end{barticle}
\endbibitem

\bibitem[\protect\citeauthoryear{Gwet and Rivest}{1992}]{gwet1992outlier}
\begin{barticle}[author]
\bauthor{\bsnm{Gwet},~\bfnm{Jean-Philippe}\binits{J.-P.}} \AND
  \bauthor{\bsnm{Rivest},~\bfnm{Louis-Paul}\binits{L.-P.}}
(\byear{1992}).
\btitle{Outlier Resistant Alternatives to the Ratio Estimator}.
\bjournal{J. Amer. Statist. Assoc.}
\bvolume{87}
\bpages{1174--1182}.
\end{barticle}
\endbibitem

\bibitem[\protect\citeauthoryear{Huber}{1973}]{huber1973robust}
\begin{barticle}[author]
\bauthor{\bsnm{Huber},~\bfnm{Peter~J}\binits{P.~J.}}
(\byear{1973}).
\btitle{Robust Regression: Asymptotics, Conjectures and Monte Carlo}.
\bjournal{Ann. Statist.}
\bpages{799--821}.
\end{barticle}
\endbibitem

\bibitem[\protect\citeauthoryear{Maechler et~al.}{2016}]{robustbase}
\begin{bmanual}[author]
\bauthor{\bsnm{Maechler},~\bfnm{Martin}\binits{M.}},
  \bauthor{\bsnm{Rousseeuw},~\bfnm{Peter}\binits{P.}},
  \bauthor{\bsnm{Croux},~\bfnm{Christophe}\binits{C.}},
  \bauthor{\bsnm{Todorov},~\bfnm{Valentin}\binits{V.}},
  \bauthor{\bsnm{Ruckstuhl},~\bfnm{Andreas}\binits{A.}},
  \bauthor{\bsnm{Salibian-Barrera},~\bfnm{Matias}\binits{M.}},
  \bauthor{\bsnm{Verbeke},~\bfnm{Tobias}\binits{T.}},
  \bauthor{\bsnm{Koller},~\bfnm{Manuel}\binits{M.}},
  \bauthor{\bsnm{Conceicao},~\bfnm{Eduardo L.~T.}\binits{E.~L.~T.}} \AND
  \bauthor{\bsnm{{Anna di Palma}},~\bfnm{Maria}\binits{M.}}
(\byear{2016}).
\btitle{robustbase: Basic Robust Statistics}
\bnote{R package version 0.92-7}.
\end{bmanual}
\endbibitem

\bibitem[\protect\citeauthoryear{Mallows}{1975}]{mallows1975some}
\begin{barticle}[author]
\bauthor{\bsnm{Mallows},~\bfnm{Colin~L}\binits{C.~L.}}
(\byear{1975}).
\btitle{On some topics in robustness}.
\bjournal{Unpublished memorandum, Bell Telephone Laboratories, Murray Hill,
  NJ}.
\end{barticle}
\endbibitem

\bibitem[\protect\citeauthoryear{Pe{\~n}a, Zamar and Yan}{2009}]{2009Pena2196}
\begin{barticle}[author]
\bauthor{\bsnm{Pe{\~n}a},~\bfnm{Daniel}\binits{D.}},
  \bauthor{\bsnm{Zamar},~\bfnm{Ruben}\binits{R.}} \AND
  \bauthor{\bsnm{Yan},~\bfnm{Guohua}\binits{G.}}
(\byear{2009}).
\btitle{Bayesian likelihood robustness in linear models}.
\bjournal{J. Statist. Plann. Inference}
\bvolume{139}
\bpages{2196-2207}.
\end{barticle}
\endbibitem

\bibitem[\protect\citeauthoryear{O'Hagan and Pericchi}{2012}]{o2012bayesian}
\begin{barticle}[author]
\bauthor{\bsnm{O'Hagan},~\bfnm{Anthony}\binits{A.}} \AND
  \bauthor{\bsnm{Pericchi},~\bfnm{Luis}\binits{L.}}
(\byear{2012}).
\btitle{Bayesian heavy-tailed models and conflict resolution: A review}.
\bjournal{Braz. J. Probab. Stat.}
\bvolume{26}
\bpages{372--401}.
\end{barticle}
\endbibitem

\bibitem[\protect\citeauthoryear{Rousseeuw and
  Yohai}{1984}]{rousseeuw1984robust}
\begin{bincollection}[author]
\bauthor{\bsnm{Rousseeuw},~\bfnm{Peter~J.}\binits{P.~J.}} \AND
  \bauthor{\bsnm{Yohai},~\bfnm{Victor~J.}\binits{V.~J.}}
(\byear{1984}).
\btitle{Robust regression by means of S-estimators}.
In \bbooktitle{Robust and Nonlinear Time Series Analysis.}
\bpages{256--272}.
\bpublisher{Springer}.
\end{bincollection}
\endbibitem

\bibitem[\protect\citeauthoryear{Salibian-Barrera and
  Yohai}{2006}]{salibian2006fast}
\begin{barticle}[author]
\bauthor{\bsnm{Salibian-Barrera},~\bfnm{Mat{\'\i}as}\binits{M.}} \AND
  \bauthor{\bsnm{Yohai},~\bfnm{V{\'\i}ctor~J}\binits{V.~J.}}
(\byear{2006}).
\btitle{A fast algorithm for {S}-regression estimates}.
\bjournal{J. Comput. Graph. Statist.}
\bvolume{15}
\bpages{414--427}.
\end{barticle}
\endbibitem

\bibitem[\protect\citeauthoryear{Scheff{\'e}}{1947}]{scheffe1947useful}
\begin{barticle}[author]
\bauthor{\bsnm{Scheff{\'e}},~\bfnm{Henry}\binits{H.}}
(\byear{1947}).
\btitle{A Useful Convergence Theorem for Probability Distributions}.
\bjournal{Ann. Math. Statist.}
\bpages{434--438}.
\end{barticle}
\endbibitem

\bibitem[\protect\citeauthoryear{Venables and Ripley}{2002}]{masspackage}
\begin{bbook}[author]
\bauthor{\bsnm{Venables},~\bfnm{W.~N.}\binits{W.~N.}} \AND
  \bauthor{\bsnm{Ripley},~\bfnm{B.~D.}\binits{B.~D.}}
(\byear{2002}).
\btitle{Modern Applied Statistics with S},
\bedition{Fourth} ed.
\bpublisher{Springer}, \baddress{New York}.
\bnote{ISBN 0-387-95457-0}.
\end{bbook}
\endbibitem

\bibitem[\protect\citeauthoryear{West}{1984}]{1984west431}
\begin{barticle}[author]
\bauthor{\bsnm{West},~\bfnm{Mike}\binits{M.}}
(\byear{1984}).
\btitle{Outlier Models and Prior Distributions in {B}ayesian Linear
  Regression}.
\bjournal{J. R. Stat. Soc. Ser. B. Stat. Methodol.}
\bvolume{46}
\bpages{431-439}.
\end{barticle}
\endbibitem

\end{thebibliography}

\section{Supplementary Material}\label{sec-supp}

Proposition~\ref{proposition-proper} and Theorem~\ref{thm-main} are proved in Sections~\ref{proof-proposition-proper} and~\ref{sec-proof-a}, respectively. The R functions that were used for the computations are provided in Section~\ref{r_functions}.

\subsection{Proofs}\label{sec-proof}

The assumptions on $f$ imply that $f(z)$ and  $|z| f(z)$ are bounded on the real line, with a limit of 0 in their tails as $|z|\rightarrow \infty$. As a result, we can define the constant $B>0$ as follows:
\begin{equation*}
B:=\max\left\{\sup_{z\in\re}f(z),\sup_{z\in\re}|z| f(z),\sup_{\beta\in\re,\sigma>0}\min(\sigma,1)\pi(\beta,\sigma)\right\}.
\end{equation*}
We also define the constant $\zeta\geq 1$ as follows:
\begin{equation*}
 \zeta:= \max_{i}\left\{\max\left\{\abs{x_i},\abs{x_i}^{-1}\right\}\right\}.
\end{equation*}
It results that for all $i\in\{1,\ldots,n\}$ and for any $\varepsilon\in\re$, we have
\begin{equation*}
 \zeta^{-|\varepsilon|}\leq \abs{x_i}^{\varepsilon}\leq \zeta^{|\varepsilon|}.
\end{equation*}
 The monotonicity of the tails of $f(z)$ and $|z| f(z)$ implies that there exists
a constant $M> 0$ such that
\begin{equation}\label{eqn-monotonic}
|y|\ge |z|\ge M\Rightarrow f(y)\le f(z) \,\text{ and }\, |y| f(y)\le |z| f(z).
\end{equation}

\subsubsection{Proof of Proposition~1}\label{proof-proposition-proper}

To prove that  $\pi(\beta,\sigma\mid\mathbf{y_n})$ is proper (the proof for $\pi(\beta,\sigma\mid\mathbf{y_k})$ is omitted because it is similar), it suffices to show that
the marginal $m(\mathbf{y_n})$ is finite. Without loss of generality, we assume for convenience that $y_1/x_1<\ldots<y_{n}/x_{n}$. Note that we have strict inequalities
 because $Y_1,\ldots,Y_n$ are continuous random variables. Let the constant $\delta>0$ be defined as
 \begin{equation*}
\delta=\zeta^{-1}\times\min_{i\in\{1,\ldots,n-1\}}\left\{\left(y_{i+1}/x_{i+1}-y_i/x_i\right)/2\right\}.
\end{equation*}
 We first show that the function is integrable on the area where the ratio $1/\sigma$ is bounded above. More precisely, we consider $\beta\in\re$ and $\delta(M\zeta^{|\theta|})^{-1}\le \sigma<\infty$. We next show that the function is integrable on the area where the ratio $1/\sigma$ approaches infinity. We have
\begin{align*}
&\int_{\delta(M\zeta^{|\theta|})^{-1}}^{\infty}\int_{-\infty}^{\infty}\pi(\beta,\sigma)
  \prod_{i=1}^{n}\sigma^{-1} |x_i|^{-\theta}f\left(\sigma^{-1} |x_i|^{-\theta}(y_i-\beta x_i)\right)\,d\beta\,d\sigma\\
&\za{\le} \max\left(\frac{1}{\sigma},1\right) B^n \zeta^{|\theta|(n-1)}\int_{\delta/(M\zeta^{|\theta|})}^{\infty}\frac{1}{\sigma^{n-1}}\int_{-\infty}^{\infty}\frac{1}{\sigma|x_1|^{\theta}}f\left(\frac{y_1-\beta x_1}{\sigma |x_1|^{\theta}}\right)\,d\beta\,d\sigma\\
&\zb{\le}\max\left(\delta^{-1}M\zeta^{|\theta|},1\right) B^n \zeta^{|\theta|(n-1)}\abs{x_1}^{-1}\int_{\delta(M\zeta^{|\theta|})^{-1}}^{\infty}\frac{1}{\sigma^{n-1}}\,d\sigma\int_{-\infty}^{\infty}f(\beta')\,d\beta'\\
&\propto \int_{\delta(M\zeta^{|\theta|})^{-1}}^{\infty}\sigma^{-(n-1)}\,d\sigma\int_{-\infty}^{\infty}f(\beta')\,d\beta'
\zc{=} (\delta^{-1}M\zeta^{|\theta|})^{n-2}(n-2)^{-1}<\infty.
\end{align*}
In step $a$, we use $\abs{x_i}^{-\theta}\leq \zeta^{|\theta|}$ for $i=2,\ldots,n$, and we bound $\min(\sigma,1)\pi(\beta,\sigma)$ and each of $n-1$ densities $f$ by $B$. In step $b$, we use the change of variable $\beta'=\sigma^{-1} |x_1|^{-\theta}(y_1-\beta x_1)$. In step $c$, we use $n>2$. Note that if instead, in step $a$, we bound $\sigma\pi(\beta,\sigma)$ by $B$, one can verify that the condition $n\ge 2$ is sufficient to bound above the integral.

 We now show that the integral is finite on $\beta\in\re$ and $0<\sigma<\delta(M\zeta^{|\theta|})^{-1}$. We have to carefully analyse the subareas where $y_i-\beta x_i$ is close to 0 in order to deal with the $0/0$ form of the ratios $(y_i-\beta x_i)/(\sigma |x_i|^{\theta})$. In order to achieve this, we split the domain of $\beta$ into $n$ mutually exclusive areas as follows: $\re=\cup_{j=1}^n \{\beta:(y_{j-1}/x_{j-1}+y_{j}/x_{j})/2\le \beta \le(y_{j}/x_{j}+y_{j+1}/x_{j+1})/2\}$, where $y_0/x_0:=-\infty$ and $y_{n+1}/x_{n+1}:=\infty$. We now consider $0<\sigma<\delta(M\zeta^{|\theta|})^{-1}$ and $(y_{j-1}/x_{j-1}+y_{j}/x_{j})/2\le \beta \le$ \, $(y_{j}/x_{j}+y_{j+1}/x_{j+1})/2$, $j\in\{1,\ldots,n\}$.
\begin{align*}
&\pi(\beta,\sigma)\prod_{i=1}^{n}\sigma^{-1} |x_i|^{-\theta}f\left(\sigma^{-1} |x_i|^{-\theta}(y_i-\beta x_i)\right)\\
&\za{\le}\sigma^{-1}B\max(1,\delta(M\zeta^{|\theta|})^{-1})\prod_{i=1}^{n}\sigma^{-1} |x_i|^{-\theta}f\left(\sigma^{-1} |x_i|^{-\theta}(y_i-\beta x_i)\right)\\
&\propto\sigma^{-1}B \sigma^{-1} |x_j|^{-\theta}f\left(\frac{y_j-\beta x_j}{\sigma|x_j|^{\theta}}\right)\prod_{i=1 (i\ne j)}^{n}\sigma^{-1} |x_i|^{-\theta}f\left(\sigma^{-1} |x_i|^{-\theta}(y_i-\beta x_i)\right)\\
&\zb{\le}\sigma^{-1}B \sigma^{-1} |x_j|^{-\theta}f\left(\sigma^{-1} |x_j|^{-\theta}(y_j-\beta x_j)\right) \left[\sigma^{-1}\zeta^{|\theta|} f\left(\sigma^{-1}\zeta^{-|\theta|}\delta\right)\right]^{n-1}\\
&\zc{\le}B^{n-1}\zeta^{|\theta|(2n-3)}\delta^{-(n-2)}\sigma^{-1} |x_j|^{-\theta}f\left((y_j-\beta x_j)/(\sigma |x_j|^{\theta})\right)\sigma^{-2} f\left(\sigma^{-1} \zeta^{-|\theta|}\delta\right)\\
&\propto\sigma^{-1} |x_j|^{-\theta}f\left(\sigma^{-1} |x_j|^{-\theta}(y_j-\beta x_j)\right)\sigma^{-2} \zeta^{-|\theta|}\delta f\left(\sigma^{-1} \zeta^{-|\theta|}\delta\right).
\end{align*}
In step $a$, we use $\pi(\beta,\sigma)\le \max(\sigma^{-1},1)B= \sigma^{-1}B \max(1,\sigma)\le \sigma^{-1}B \max(1,$ $\delta(M\zeta^{|\theta|})^{-1})$.
In step $b$, for $i\ne j$, we first note that
 \begin{align*}
   &|y_i-\beta x_i|=|x_i||y_i/x_i-\beta|\ge \zeta^{-1}|y_i/x_i-\beta|\\
   &\qquad\ge \zeta^{-1}\times\min\left\{(y_{j}/x_{j}-y_{j-1}/x_{j-1})/2,(y_{j+1}/x_{j+1}-y_{j}/x_{j})/2\right\}\ge \delta,
\end{align*}
and then we use $f(\sigma^{-1} |x_i|^{-\theta}(y_i-\beta x_i))\le f(\sigma^{-1} \zeta^{-|\theta|}\delta)$ by the monotonicity of the tails of $f$ since
 $\sigma^{-1} |x_i|^{-\theta}|y_i-\beta x_i|\ge \sigma^{-1} |x_i|^{-\theta}\delta\ge  \sigma^{-1}\zeta^{-|\theta|}\delta \ge
  \delta^{-1} M\zeta^{|\theta|} \zeta^{-|\theta|}\delta=M$. Again for $i\neq j$, we use $\abs{x_i}^{-\theta}\leq \zeta^{|\theta|}$.
   In step $c$, we bound $n-2$ terms $\sigma^{-1}f(\sigma^{-1} \zeta^{-|\theta|}\delta)$  by $\zeta^{|\theta|}\delta^{-1}B$.

Finally, we have
\begin{align*}
\int_{0}^{\delta(M\zeta^{|\theta|})^{-1}}\frac{\delta}{\sigma^{2}\zeta^{|\theta|}}&  f\left(\frac{\delta}{\sigma \zeta^{|\theta|}}\right)
\int_{(y_{j-1}/x_{j-1}+y_{j}/x_{j})/2}^{(y_{j}/x_{j}+y_{j+1}/x_{j+1})/2}\frac{1}{\sigma |x_j|^{\theta}}f\left(\frac{y_j-\beta x_j}{\sigma |x_j|^{\theta}}\right)\,d\beta\,d\sigma\\
&\le \abs{x_j}^{-1}\int_{0}^{\infty}f(\sigma')\,d\sigma'\int_{-\infty}^{\infty}f(\beta')\,d\beta'\le\abs{x_j}^{-1}\le \zeta<\infty,
\end{align*}
where we use the change of variables  $\sigma'=\sigma^{-1} \zeta^{-|\theta|}\delta$ and $\beta'=\sigma^{-1} |x_j|^{-\theta}(y_j-\beta x_j)$. Note that we do not need to assume that $f$ is a log-regularly varying distribution to obtain the result.

\subsubsection{Proof of Theorem~1}\label{sec-proof-a}

Consider the model and the context described in Section~\ref{sec-model}. We assume that
$z f(z)\in L_{\rho}(\infty)$ and $k>\max(m,p)$. In addition, we assume that $m+p\geq 1$, i.e.\ that there is at least one outlier, otherwise the proof would be trivial. Two lemmas are first given and the proofs of results~(a) to (e) follows. The proofs of these two lemmas can be found in \cite{desgagne2015robustness}.

\begin{Lemma}\label{cor-location-scale-transformation}
$\forall\lambda\ge 0$, $\forall\tau\ge 1$, there exists a constant $D(\lambda,\tau)\ge 1$  such that $z\in\re$  and
$(\mu,\sigma)\in [-\lambda,\lambda]\times[1/\tau, \tau]\Rightarrow$
\begin{equation*}
1/D(\lambda,\tau)\le (1/\sigma) f((z-\mu)/\sigma)/f(z)\le D(\lambda,\tau).
\end{equation*}
\end{Lemma}
  Note that Lemma~\ref{cor-location-scale-transformation} is a corollary of Proposition~4 of \cite{desgagne2015robustness}.

\begin{Lemma}\label{lem-convolution1}
There exists a constant $C>0$ such that
\begin{equation*}
|z|\ge 2 M\Rightarrow\sup_{\mu\in\re}\frac{f(\mu)f(z-\mu)}{f(z)}\le C,
\end{equation*}
where $M$ is given in equation~(\ref{eqn-monotonic}).
\end{Lemma}

\begin{proof}[Proof of Result~(a)]
We first observe that
\begin{align*}
   &\frac{m(\mathbf{y_n})}{m(\mathbf{y_k})\prod_{i=1}^{n}[f(y_i)]^{m_i+p_i}}\\
   &= \frac{m(\mathbf{y_n})}{m(\mathbf{y_k})\prod_{i=1}^{n}[f(y_i)]^{m_i+p_i}}
     \int_{-\infty}^{\infty}\int_{0}^{\infty}\pi(\beta,\sigma\mid \mathbf{y_n})\,d\sigma\,d\beta\\
   &=   \int_{-\infty}^{\infty}\int_{0}^{\infty}\frac{\pi(\beta,\sigma)\prod_{i=1}^{n}
        \left[\sigma^{-1}|x_i|^{-\theta}f\left(\sigma^{-1}|x_i|^{-\theta}(y_i-\beta x_i)\right)\right]^{k_i+m_i+p_i}}{m(\mathbf{y_k})\prod_{i=1}^{n}[f(y_i)]^{m_i+p_i}}\,d\sigma\,d\beta\\
   &=   \int_{-\infty}^{\infty}\int_{0}^{\infty}
        \pi(\beta,\sigma\mid \mathbf{y_k})  \prod_{i=1}^{n}\left[\frac{\sigma^{-1}|x_i|^{-\theta}f\left(\sigma^{-1}|x_i|^{-\theta}(y_i-\beta x_i)\right)}{f(y_i)}\right]^{m_i+p_i}\,d\sigma\,d\beta.
\end{align*}
We show that the last integral converges to 1 as $\omega\rightarrow\infty$ to prove result~(a).
If we use Lebesgue's dominated convergence theorem to interchange the limit $\omega\rightarrow\infty$ and the integral, we have
\begin{align*}
\lim_{\omega\rightarrow\infty}&\int_{-\infty}^{\infty}\int_{0}^{\infty}\pi(\beta,\sigma\mid \mathbf{y_k})
 \prod_{i=1}^{n}\left[\frac{\frac{1}{\sigma|x_i|^{\theta}}f\left(\frac{y_i-\beta x_i}{\sigma|x_i|^{\theta}}\right)}{f(y_i)}\right]^{m_i+p_i}\,d\sigma\,d\beta\\
 & = \int_{-\infty}^{\infty}\int_{0}^{\infty}\lim_{\omega\rightarrow\infty}
  \pi(\beta,\sigma\mid \mathbf{y_k})\prod_{i=1}^{n}\left[\frac{\frac{1}{\sigma|x_i|^{\theta}}f\left(\frac{y_i-\beta x_i}{\sigma|x_i|^{\theta}}\right)}{f(y_i)}\right]^{m_i+p_i}
  \,d\sigma\,d\beta \\
 & = \int_{-\infty}^{\infty}\int_{0}^{\infty}
  \pi(\beta,\sigma\mid \mathbf{y_k})
  \,d\sigma\,d\beta = 1,
\end{align*}
using Proposition~4 of \cite{desgagne2015robustness} in the second equality, since $x_1,\ldots,$ $x_n$ and $\theta$ are fixed, and then Proposition~\ref{proposition-proper}. Note that pointwise convergence is sufficient, for any value of $\beta\in\re$ and $\sigma>0$, once the limit is inside the integral. However, in order to use Lebesgue's dominated convergence theorem, we need to  show that the integrand is bounded, for any value of $\omega\ge \yo$, by an integrable function of $\beta$ and $\sigma$ that does not depend on $\omega$. The constant $\yo$ can be chosen as large as we want, and minimum values for $\yo$ will be given throughout the proof. In order to bound the integrand, we divide the domain of integration into four quadrants delineated by the axes $\beta=0$ and $\sigma=1$. The proofs are given only for the two quadrants where $\beta\ge 0$ because the proofs for $\beta<0$ are similar. The strategy is again to separately analyse the area where the ratio $1/\sigma$ approaches infinity.

We assumed that $y_i$ can be written as $y_i=a_i+b_i \omega$, where $\omega\rightarrow\infty$, $a_i$ and $b_i$ are constants such that $a_i\in\re$ and $b_i\ne 0$ if $y_i$ is an outlier. Therefore, the ranking of the elements in the set $\{|y_i| : m_i+p_i=1\}$ is primarily determined by the values $|b_1|,\ldots,|b_n|$ and we can choose the constant $\yo$ larger than a certain threshold such that this ranking remains unchanged for all $\omega\ge \yo$. Without loss of generality, we assume for convenience that
\begin{equation*}
\omega=\min_{\{i\,:\,m_i+p_i=1\}}|y_i|\hspace{5mm}\text{ and consequently}\hspace{5mm} \min_{\{i\,:\,m_i+p_i=1\}} |b_i|=1,
\end{equation*}
and we also assume that $y_1$ is a nonoutlier (therefore $k_1=1$). We now bound above the integrand on the first quadrant.

\textbf{Quadrant~1:} Consider $0\le\beta<\infty$ and $1\le\sigma<\infty$. We have
\begin{align*}
&\pi(\beta,\sigma\mid \mathbf{y_k})  \prod_{i=1}^{n}\left[\frac{\sigma^{-1}|x_i|^{-\theta}f\left(\sigma^{-1}|x_i|^{-\theta}(y_i-\beta x_i)\right)}{f(y_i)}\right]^{m_i+p_i}\\
&\propto\frac{\pi(\beta,\sigma)}{\sigma^{n}}\prod_{i=1}^{n}\frac{|x_i|^{-\theta}f\left(\sigma^{-1}|x_i|^{-\theta}(y_i-\beta x_i)\right)}{\left[f(y_i)\right]^{m_i+p_i}}\\
&\za{\le}\frac{B}{\sigma^{n}}\prod_{i=1}^{n} \frac{D(|a_i|,\zeta^{|\theta|})f((b_i \omega-\beta x_i)/\sigma)}{\left[f(y_i)\right]^{m_i+p_i}}\\
&\zb{\le}\frac{1}{[f(\omega)]^{m+p}}\frac{B}{\sigma^{n}}\prod_{i=1}^{n} D(|a_i|,\zeta^{|\theta|})f((b_i \omega-\beta x_i)/\sigma)\left[|b_i|D(|a_i|,|b_i|)\right]^{m_i+p_i}\\
&\propto\frac{1}{[f(\omega)]^{m+p}}\frac{1}{\sigma^{n}}\prod_{i=1}^{n} f((b_i \omega-\beta x_i)/\sigma)\\
&\zc{=}\frac{1}{[f(\omega)]^{m+p}}\frac{1}{\sigma^{n}}
  \prod_{i=1}^{n}[f(\beta x_i/\sigma)]^{k_i}\left[f((b_i \omega-\beta x_i)/\sigma)\right]^{m_i+p_i}\\
  &\zd{=} \frac{\frac{1}{\sigma}f\left(\frac{\beta x_1}{\sigma}\right)}{\sigma^{k-3/2}}\left[\frac{\omega/\sigma}{\omega f(\omega)}\right]^{m+p}\frac{1}{\sigma^{1/2}}
  \prod_{i=2}^{n}\left[f\left(\frac{\beta x_i}{\sigma}\right)\right]^{k_i}\left[f\left(\frac{b_i \omega-\beta x_i}{\sigma}\right)\right]^{m_i+p_i}.
\end{align*}
In step $a$, we use $y_i=a_i+b_i \omega$ and Lemma~\ref{cor-location-scale-transformation} to obtain
\begin{align*}
\frac{1}{|x_i|^{\theta}}f\left(\frac{y_i-\beta x_i}{\sigma |x_i|^{\theta}}\right)&=\frac{1}{|x_i|^{\theta}}f\left(\frac{(b_i \omega-\beta x_i)/\sigma+a_i/\sigma}{|x_i|^{\theta}}\right)\cr
&\le D(|a_i|,\zeta^{|\theta|})f\left(\frac{b_i \omega-\beta x_i}{\sigma}\right)
\end{align*}
because $|a_i/\sigma|\le |a_i|$ and $\zeta^{-|\theta|}\le |x_i|^{\theta}\leq \zeta^{|\theta|}$, for all $i$. We also use $\pi(\beta,\sigma)\le \max(\sigma^{-1},1)B= B$.
In step $b$, we again use Lemma~\ref{cor-location-scale-transformation} to obtain $f(\omega)/f(y_i)=f((y_i-a_i)/b_i)/f(y_i)\le |b_i|D(|a_i|,|b_i|)$.
In step $c$, we set $b_i=0$ if $k_i=1$ and we use the symmetry of $f$ to obtain $f(-\beta x_i/\sigma)=f(\beta x_i/\sigma)$.
In step $d$, we use the assumption $k_1=1$, which implies that $m_1=p_1=0$.

Now it suffices to demonstrate that
\begin{equation}\label{fct_quadrant1}
\left[\frac{\omega/\sigma}{\omega f(\omega)}\right]^{m+p}\frac{1}{\sigma^{1/2}}
  \prod_{i=2}^{n}[f(\beta x_i/\sigma)]^{k_i}\left[f((b_i \omega-\beta x_i)/\sigma)\right]^{m_i+p_i}
  \end{equation}
is bounded by a constant that does not depend on $\omega,\beta$ and $\sigma$ since $(1/\sigma)^{k-3/2}$ $\times(1/\sigma)f(\beta x_1/\sigma)$ is an integrable function on quadrant~1.
Indeed, since $k>2$, we have
\begin{align*}
\int_{1}^{\infty}\frac{1}{\sigma^{k-3/2}}\int_{0}^{\infty}\frac{1}{\sigma}f\left(\frac{\beta x_1}{\sigma}\right) \,d\beta\,d\sigma
\le  \frac{1}{\abs{x_1}}\int_{1}^{\infty}\frac{1}{\sigma^{k-3/2}}\,d\sigma=\frac{\abs{x_1}^{-1}}{k-5/2}\le 2\zeta.
\end{align*}
Note that if instead, in step $a$, we bound $\pi(\beta,\sigma)$ by $\sigma^{-1}B$, one can verify that the condition $k\ge 2$ is sufficient to bound above the integral.

In order to bound above the function in (\ref{fct_quadrant1}), we separately analyse the three following cases: $\omega/\sigma$ is large, $\omega/\sigma$ is either large or bounded, and $\omega/\sigma$ is bounded. More precisely, we split quadrant~1 with respect to $\sigma$ into three parts: $1\leq \sigma<\omega^{1/2}$, $\omega^{1/2}\leq \sigma< \omega/(2M)$ and $\omega/(2M)\leq \sigma<\infty$, where $M$ is defined in equation~(\ref{eqn-monotonic}). Note that this is well defined if $\yo> \max(1,(2M)^2)$ since $\omega\ge \yo$.

First, we consider $0\le\beta<\infty$ and $\omega/(2 M)\le\sigma<\infty$. We have,
\begin{align*}
&\left[\frac{\omega/\sigma}{\omega f(\omega)}\right]^{m+p}\frac{1}{\sigma^{1/2}}
  \prod_{i=2}^{n}[f(\beta x_i/\sigma)]^{k_i}\left[f((b_i \omega-\beta x_i)/\sigma)\right]^{m_i+p_i} \cr
&\za{\le} B^{n-1}\left[\frac{\omega/\sigma}{\omega f(\omega)}\right]^{m+p}\frac{1}{\sigma^{1/2}}\zb{\le}B^{n-1}(2M)^{m+p+1/2}\frac{(1/\omega)^{1/2}}{[\omega f(\omega)]^{m+p}}\cr
&\zc{\le}B^{n-1}(2M)^{m+p+1/2}\frac{(1/\omega)^{1/2}}{{(\log \omega)^{-(\rho+1)(m+p)}}}\\
&\zd{\le}B^{n-1}(2M)^{m+p+1/2}[2(\rho+1)(m+p)/\ee]^{(\rho+1)(m+p)}<\infty.
\end{align*}
In step $a$, we use $f\le B$.
In step $b$, we use $\omega/\sigma\le 2M$ and $(1/\sigma)\le (2M)/\omega$.
In step $c$, we use $\omega f(\omega)>(\log \omega)^{-\rho-1}$ if $\omega\ge \yo\ge A(1)$, where $A(1)$ comes from
    Proposition~2 of \cite{desgagne2015robustness}.
For step $d$, it is purely algebraic to show that the maximum of $(\log \omega)^{\xi}/\omega^{1/2}$ is $(2\xi/\ee)^{\xi}$
      for $\omega>1$ and $\xi>0$, where $\xi=(\rho+1)(m+p)$ in our situation.

Now, consider the two other parts combined (we will split them in the next step), that is $0\le\beta<\infty$ and $1\le\sigma\le \omega/(2M)$.
 We have,
\begin{align*}
&\left[\frac{\omega/\sigma}{\omega f(\omega)}\right]^{m+p}\frac{1}{\sigma^{1/2}}
  \prod_{i=2}^{n}[f(\beta x_i/\sigma)]^{k_i}\left[f((b_i \omega-\beta x_i)/\sigma)\right]^{m_i+p_i}\\
  &\hspace{0mm}\za{\le}\left[\frac{\omega/\sigma}{\omega f(\omega)}\right]^{m+p}\frac{1}{\sigma^{1/2}}
  \prod_{i=2}^{n}[f(\beta x_i/\sigma)]^{k_i}\left[f(b_i \omega/\sigma)\right]^{m_i}\left[f((b_i \omega-\beta x_i)/\sigma)\right]^{p_i}\\
  &\hspace{0mm}=\left[\frac{\omega/\sigma}{\omega f(\omega)}\right]^{m+p}\frac{1}{\sigma^{1/2}}
  \prod_{i=2}^{n}\left[f\left(\frac{\beta x_i}{\sigma}\right)\right]^{k_i-p_i}\left[f\left(\frac{b_i \omega}{\sigma}\right)\right]^{m_i+p_i} \cr
  &\qquad\times\left[\frac{f\left(\frac{b_i \omega}{\sigma}-\frac{\beta x_i}{\sigma}\right)f\left(\frac{\beta x_i}{\sigma}\right)}{f\left(\frac{b_i\omega}{\sigma}\right)}\right]^{p_i}\\
&\hspace{0mm}\zb{\le}C^p \frac{1}{\sigma^{1/2}}
  \prod_{i=2}^{n}[f(\beta x_i/\sigma)]^{k_i-p_i}\left[\frac{(\omega/\sigma)f(b_i \omega/\sigma)}{\omega f(\omega)}\right]^{m_i+p_i}\\
&\hspace{0mm}\zc{\le}C^p \frac{1}{\sigma^{1/2}}
  \left[\frac{(\omega/\sigma)f(\omega/\sigma)}{\omega f(\omega)}\right]^{m+p}\prod_{i=2}^{n}[f(\beta x_i/\sigma)]^{k_i-p_i}\\
&\hspace{0mm}\zd{\le}C^p \frac{1}{\sigma^{1/2}}
  \left[\frac{(\omega/\sigma)f(\omega/\sigma)}{\omega f(\omega)}\right]^{m+p}[f(\beta /\sigma)]^{k-1-p}[\zeta D(0,\zeta)]^{k-1+p}\\
   &\hspace{0mm}\ze{\le}C^p \frac{1}{\sigma^{1/2}}  \left[\frac{(\omega/\sigma)f(\omega/\sigma)}{\omega f(\omega)}\right]^{m+p}B^{k-1-p}[\zeta D(0,\zeta)]^{k-1+p} \cr
   &\propto \frac{1}{\sigma^{1/2}}\left[\frac{(\omega/\sigma)f(\omega/\sigma)}{\omega f(\omega)}\right]^{m+p}.
\end{align*}
In step $a$, we use $f((b_i \omega-\beta x_i)/\sigma)\le f(b_i \omega/\sigma)$ if $m_i=1$ (in this case $x_i>0, b_i<0$ or $x_i<0, b_i>0$) by the monotonicity of the tails of $f$ since $|b_i \omega-\beta x_i|/\sigma=(|b_i| \omega+\beta |x_i|)/\sigma\ge |b_i| \omega/\sigma$ $ \ge |b_i| (2M) \ge 2M \ge M$.
 In step $b$,
   we use Lemma~\ref{lem-convolution1} since $|b_i| \omega/\sigma \ge |b_i| (2M) \ge 2M$.
In step $c$, we use
   $f(b_i \omega/\sigma)\le f(\omega/\sigma)$ by the monotonicity of the tails of $f$  since $|b_i|\omega/\sigma\ge \omega/\sigma\ge 2M \ge M$.
In step $d$, we use Lemma~\ref{cor-location-scale-transformation} to obtain $f(\beta \abs{x_i}/\sigma)\le |x_i|^{-1}D(0,\zeta)f(\beta /\sigma)\le
\zeta D(0,\zeta)f(\beta /\sigma)$, and similarly $1/f(\beta \abs{x_i}/\sigma)\le \zeta D(0,\zeta)/f(\beta /\sigma)$.
In step $e$, we use $[f(\beta /\sigma)]^{k-1-p}\le B^{k-1-p}$ since $k-1\geq p$ (by assumption $k>\max(m,p)\Rightarrow k>p$).

Now, we consider $0\le\beta<\infty$ and $\omega^{1/2}\le\sigma\le \omega/(2M)$. We have,
\begin{align*}
\frac{1}{\sigma^{1/2}}
  \left[\frac{(\omega/\sigma)f(\omega/\sigma)}{\omega f(\omega)}\right]^{m+p}& \za{\le} B^{m+p}
 \frac{(1/\omega)^{1/4}}{[\omega f(\omega)]^{m+p}}\zb{\le} B^{m+p}
 \frac{(1/\omega)^{1/4}}{(\log \omega)^{-(\rho+1)(m+p)}}\\
&\zc{\le} B^{m+p}[4(\rho+1)(m+p)/\ee]^{(\rho+1)(m+p)}<\infty.
\end{align*}
In step $a$, we use $(\omega/\sigma)f(\omega/\sigma)\le B$ and $(1/\sigma)^{1/2}\le (1/\omega)^{1/4}$.
In step $b$, we use $\omega f(\omega)>(\log \omega)^{-\rho-1}$ if $\omega\ge \yo\ge A(1)$, where $A(1)$ comes from
    Proposition~2 of \cite{desgagne2015robustness}.
In step $c$, it is purely algebraic to show that the maximum of $(\log \omega)^{\xi}/\omega^{1/4}$ is $(4\xi/\ee)^{\xi}$
      for $\omega>1$ and $\xi>0$, where $\xi=(\rho+1)(m+p)$ in our situation.

Finally, we consider $0\le\beta<\infty$ and $1\le\sigma\le \omega^{1/2}$. We have,
\begin{equation*}
\frac{1}{\sigma^{1/2}}
  \left[\frac{(\omega/\sigma)f(\omega/\sigma)}{\omega f(\omega)}\right]^{m+p}\za{\le}
  \left[\frac{\omega^{1/2}f(\omega^{1/2})}{\omega f(\omega)}\right]^{m+p}\zb{\le} 2^{(\rho+1)(m+p)}<\infty.
\end{equation*}
In step $a$, we use $1/\sigma\le 1$ and we use $(\omega/\sigma)f(\omega/\sigma)\le \omega^{1/2} f(\omega^{1/2})$ by the monotonicity of the tails of
$|z| f(z)$ since $\omega/\sigma\ge \omega^{1/2}\ge \yo^{1/2}\ge M$ if $\yo\ge M^2$. In step $b$, we use $\omega^{1/2} f(\omega^{1/2})/(\omega
f(\omega))$ $\le 2(1/2)^{-\rho}=2^{\rho+1}$
if $\omega\ge \yo\ge A(1,2)$, where $A(1,2)$ comes from the definition of a log-regularly varying function (see Definition~1 of \cite{desgagne2015robustness}).

\textbf{Quadrant~2:} Consider $-\infty<\beta< 0$ and $1\le\sigma<\infty$.
The proof for quadrant~2 is similar to that of quadrant~1.
The condition $k> p$ is replaced by $k> m$. Note that $k>\max(m,p)$ is assumed in Theorem~1.

\textbf{Quadrant~3:} Consider $-\infty<\beta< 0$ and $0<\sigma< 1$.
The proof for quadrant~3 is similar to that of quadrant~4, given below.
The condition $k>p$ is replaced by $k>m$. Note that $k>\max(m,p)$ is assumed in Theorem~1.

\textbf{Quadrant~4:} Consider $0\le\beta<\infty$ and $0<\sigma< 1$.
We actually need to show that
\begin{align*}
  \lim_{\omega\rightarrow\infty} \int_{0}^{\infty}\int_{0}^{1}
  \pi(\beta,\sigma\mid \mathbf{y_k})&\prod_{i=1}^{n}\left[\frac{\frac{1}{\sigma|x_i|^{\theta}}f\left(\frac{y_i-\beta x_i}{\sigma|x_i|^{\theta}}\right)}{f(y_i)}\right]^{m_i+p_i}
  \,d\sigma\,d\beta \\& = \int_{0}^{\infty}\int_{0}^{1}\pi(\beta,\sigma\mid \mathbf{y_k})\,d\sigma\,d\beta.
\end{align*}
For quadrant~4, we proceed in a slightly different manner than for quadrant~1. We begin by separating the first integral into two parts as follows:
\small\begin{align*}
  &\lim_{\omega\rightarrow\infty} \int_{0}^{\infty}\int_{0}^{1}
  \pi(\beta,\sigma\mid \mathbf{y_k})\prod_{i=1}^{n}\left[\frac{\frac{1}{\sigma|x_i|^{\theta}}f\left(\frac{y_i-\beta x_i}{\sigma|x_i|^{\theta}}\right)}{f(y_i)}\right]^{m_i+p_i}
  \,d\sigma\,d\beta \\
  &=\lim_{\omega\rightarrow\infty} \int_{0}^{\infty}\int_{0}^{1}
  \pi(\beta,\sigma\mid \mathbf{y_k})\prod_{i=1}^{n}\left[\frac{\frac{1}{\sigma|x_i|^{\theta}}f\left(\frac{y_i-\beta x_i}{\sigma|x_i|^{\theta}}\right)}{f(y_i)}\right]^{m_i+p_i}
  \ind_{[0,\zeta^{-1}\omega/2]}(\beta)\,d\sigma\,d\beta \\
  &+\lim_{\omega\rightarrow\infty} \int_{\zeta^{-1}\omega/2}^{\infty}\int_{0}^{1}
  \pi(\beta,\sigma\mid \mathbf{y_k})\prod_{i=1}^{n}\left[\frac{\frac{1}{\sigma|x_i|^{\theta}}f\left(\frac{y_i-\beta x_i}{\sigma|x_i|^{\theta}}\right)}{f(y_i)}\right]^{m_i+p_i}
  \,d\sigma\,d\beta,
\end{align*}\normalsize
where the indicator function $\ind_{A}(\beta)$ is equal to 1 if $\beta\in A$, and equal to 0 otherwise. We show that
the first part  is equal to the integral $\int_{0}^{\infty}\int_{0}^{1}\pi(\beta,\sigma\mid \mathbf{y_k})\,d\sigma\,d\beta$
and the second part is equal to 0.

For the first part, we again use Lebesgue's dominated convergence theorem in order to interchange the limit $\omega\rightarrow\infty$ and the integral. We have
\small\begin{align*}
&\lim_{\omega\rightarrow\infty}\int_{0}^{\infty}\int_{0}^{1}\pi(\beta,\sigma\mid \mathbf{y_k})  \prod_{i=1}^{n}\left[\frac{\frac{1}{\sigma|x_i|^{\theta}}f\left(\frac{y_i-\beta x_i
}{\sigma|x_i|^{\theta}}\right)}{f(y_i)}\right]^{m_i+p_i}
\ind_{[0,\zeta^{-1}\omega/2]}(\beta)\,d\sigma\,d\beta\\
 & = \int_{0}^{\infty}\int_{0}^{1}
  \pi(\beta,\sigma\mid \mathbf{y_k})\lim_{\omega\rightarrow\infty}\prod_{i=1}^{n}\left[\frac{\frac{1}{\sigma|x_i|^{\theta}}f\left(\frac{y_i-\beta x_i
}{\sigma|x_i|^{\theta}}\right)}{f(y_i)}\right]^{m_i+p_i}\ind_{[0,\zeta^{-1}\omega/2]}(\beta)
  \,d\sigma\,d\beta \\
 & =\int_{0}^{\infty}\int_{0}^{1}
  \pi(\beta,\sigma\mid \mathbf{y_k})\times 1 \times \ind_{[0,\infty)}(\beta)
  \,d\sigma\,d\beta  =\int_{0}^{\infty}\int_{0}^{1}  \pi(\beta,\sigma\mid \mathbf{y_k})\,d\sigma\,d\beta,
\end{align*}\normalsize
using Proposition~4 of \cite{desgagne2015robustness} in the second equality since $x_1,\ldots,x_n$ and $\theta$ are fixed.
Note that pointwise convergence is sufficient, for any value of $\beta\in\re$ and $\sigma>0$,
once the limit is inside the integral. We now demonstrate that the integrand is bounded, for any value of $\omega\ge \yo$, by an integrable function of $\beta$ and $\sigma$ that does not depend on $\omega$.

Consider  $0\le\beta\le \zeta^{-1}\omega/2$ (the integrand is equal to 0 if $\zeta^{-1}\omega/2< \beta<\infty$) and $0<\sigma< 1$. We have
\begin{align*}
\pi(\beta,\sigma\mid \mathbf{y_k})&
\prod_{i=1}^{n}\left[\frac{\sigma^{-1}|x_i|^{-\theta}f\left(\sigma^{-1}|x_i|^{-\theta}(y_i-\beta x_i)\right)}{f(y_i)}\right]^{m_i+p_i}\ind_{[0,\zeta^{-1}\omega/2]}(\beta)\\
&\qquad\za{\le}\pi(\beta,\sigma\mid \mathbf{y_k})
\prod_{i=1}^{n}\left[\frac{\zeta^{-|\theta|}f\left(\zeta^{-|\theta|}(y_i-\beta x_i)\right)}{f(y_i)}\right]^{m_i+p_i}\\
&\qquad\zb{\le}\pi(\beta,\sigma\mid \mathbf{y_k})
\prod_{i=1}^{n}\left[\frac{\zeta^{-|\theta|}f\left(\zeta^{-|\theta|}\omega/2\right)}{f(y_i)}\right]^{m_i+p_i}\\
&\qquad\zc{\le}\pi(\beta,\sigma\mid \mathbf{y_k})
\prod_{i=1}^{n}\left[2|b_i| D(|a_i|,2|b_i|\zeta^{|\theta|}\right]^{m_i+p_i},
\end{align*}
and $\pi(\beta,\sigma\mid \mathbf{y_k})$ is an integrable function. In step $a$, we use the equality $\ind_{[0,\zeta^{-1}\omega/2]}(\beta)= 1$. We also use
\begin{equation*}
 \frac{|y_i-\beta x_i|}{\sigma|x_i|^{\theta}}f\left(\frac{y_i-\beta x_i}{\sigma|x_i|^{\theta}}\right)\le \zeta^{-|\theta|}|y_i-\beta x_i|f\left(\zeta^{-|\theta|}(y_i-\beta x_i)\right)
\end{equation*}
by the monotonicity of the tails of $|z| f(z)$ and therefore we obtain
\begin{equation*}
 \sigma^{-1}|x_i|^{-\theta}f\left(\sigma^{-1}|x_i|^{-\theta}(y_i-\beta x_i)\right)\le \zeta^{-|\theta|}f\left(\zeta^{-|\theta|}(y_i-\beta x_i)\right),
\end{equation*}
and in step $b$, we use
\begin{equation*}
f\left(\zeta^{-\theta}(y_i-\beta x_i)\right)\le f(\zeta^{-|\theta|}\omega/2)
\end{equation*}
by the monotonicity of the tails of $f(z)$. Indeed, if $m_i=1$ (in this case $x_i>0, b_i<0$ or $x_i<0, b_i>0$), we have
 $\sigma^{-1}|x_i|^{-\theta}|y_i-\beta x_i|\ge |x_i|^{-\theta}|y_i-\beta x_i|\ge \zeta^{-|\theta|}|y_i-\beta x_i|=\zeta^{-|\theta|}(\abs{y_i}+\beta\abs{x_i})\ge \zeta^{-|\theta|}\abs{y_i}\ge  \zeta^{-|\theta|} \omega\ge  \zeta^{-|\theta|} \omega/2$ $\ge  \zeta^{-|\theta|} \yo/2\ge M$, if we choose $\yo\ge 2\zeta^{|\theta|} M$.
 And, if $p_i=1$, we have
 $\sigma^{-1}|x_i|^{-\theta}|y_i-\beta x_i|\ge \zeta^{-|\theta|}|y_i-\beta x_i|\ge \zeta^{-|\theta|}(\abs{y_i}-\beta\abs{x_i})\ge \zeta^{-|\theta|}(\omega-(\zeta^{-1}\omega/2)\zeta)=\zeta^{-|\theta|}\omega/2\ge  \zeta^{-|\theta|} \yo/2\ge M.$
Note that $0\le\beta\le \zeta^{-1}\omega/2$ is used only for the case $p_i=1$ ($\beta\ge 0$ is sufficient for the case $m_i=1$).
In step c, we use Lemma~\ref{cor-location-scale-transformation} to obtain
\begin{equation*}
\frac{f(\zeta^{-|\theta|}\omega/2)}{f(y_i)}=\frac{f((y_i-a_i)/(2b_i\zeta^{|\theta|}))}{f(y_i)}\le 2|b_i|\zeta^{|\theta|} D(|a_i|,2|b_i|\zeta^{|\theta|}).
\end{equation*}
We now prove that
\begin{equation*}
\lim_{\omega\rightarrow\infty} \int_{\zeta^{-1}\omega/2}^{\infty}\int_{0}^{1}
  \pi(\beta,\sigma\mid \mathbf{y_k})\prod_{i=1}^{n}\left[\frac{\frac{1}{\sigma|x_i|^{\theta}}f\left(\frac{y_i-\beta x_i}{\sigma|x_i|^{\theta}}\right)}{f(y_i)}\right]^{m_i+p_i}
  \,d\sigma\,d\beta=0.
  \end{equation*}
We first bound above the integrand and then we prove that the integral of the upper bound converges towards 0 as $\omega\rightarrow\infty$.

Consider  $\zeta^{-1}\omega/2< \beta <\infty$ and $0<\sigma< 1$. We have
\small\begin{align*}
&\pi(\beta,\sigma\mid \mathbf{y_k})\prod_{i=1}^{n}\left[\frac{\sigma^{-1}|x_i|^{-\theta}f\left(\sigma^{-1}|x_i|^{-\theta}(y_i-\beta x_i)\right)}{f(y_i)}\right]^{m_i+p_i}\\
&\hspace{0mm}\za{\le}\pi(\beta,\sigma\mid \mathbf{y_k})  \prod_{i=1}^{n}[2|b_i|   D(|a_i|,2|b_i|\zeta^{|\theta|})]^{m_i}
\left[\frac{|b_i|D(|a_i|,|b_i|)\frac{1}{\sigma|x_i|^{\theta}}f\left(\frac{y_i-\beta x_i}{\sigma|x_i|^{\theta}}\right)}{f(\omega)}\right]^{p_i}\\
&\hspace{0mm}\propto \pi(\beta,\sigma)\prod_{i=1}^{n}\left[\sigma^{-1} |x_i|^{-\theta}f\left(\sigma^{-1} |x_i|^{-\theta}(a_i-\beta x_i)\right)\right]^{k_i}\left[\frac{\frac{1}{\sigma|x_i|^{\theta}}f\left(\frac{y_i-\beta x_i}{\sigma|x_i|^{\theta}}\right)}{f(\omega)}\right]^{p_i}\\
&\hspace{0mm}\zb{\le}\sigma^{-1}B\left[4\zeta^{2|\theta|+2} D(0,4\zeta^{2+|\theta|}) (1/\sigma)f(\omega/\sigma)\right]^{k}\prod_{i=1}^{n} \left[\frac{\frac{1}{\sigma|x_i|^{\theta}}f\left(\frac{y_i-\beta x_i}{\sigma|x_i|^{\theta}}\right)}{f(\omega)}\right]^{p_i}\\
&\hspace{0mm}\propto \sigma^{-1}\left[\sigma^{-1}f(\sigma^{-1}\omega)\right]^{k}\prod_{i=1}^{n}\left[\frac{\sigma^{-1}|x_i|^{-\theta}f\left(\sigma^{-1}|x_i|^{-\theta}(y_i-\beta x_i)\right)}{f(\omega)}\right]^{p_i} \\
&\hspace{0mm}\zc{\le} \sigma^{-1}\left[\sigma^{-1}f(\sigma^{-1}\omega)\right]^{k-p}\prod_{i=1}^{n}\left[\sigma^{-1}|x_i|^{-\theta}f\left(\sigma^{-1}|x_i|^{-\theta}(y_i-\beta x_i)\right)\right]^{p_i} \\
&\hspace{0mm}\zd{=} \sigma^{-1}\left[\sigma^{-1}f(\sigma^{-1}\omega)\right]^{k-p}\prod_{i=1}^{p}\sigma^{-1}|x_i|^{-\theta}f\left(\sigma^{-1}|x_i|^{-\theta}(y_i-\beta x_i)\right).
\end{align*}\normalsize
In step $a$, for the case $m_i=1$, we use the inequality
$\sigma^{-1}|x_i|^{-\theta}f(\sigma^{-1}|x_i|^{-\theta}(y_i$ $-\beta x_i))/f(y_i) \leq  2|b_i| D(|a_i|,2|b_i|\zeta^{|\theta|})$
 by the same arguments used for the first part (steps $a$ to $c$). Note that we still have $\beta\ge 0$ ($0\le\beta\le \zeta^{-1}\omega/2$ was used only for the case $p_i=1$).
 For the case $p_i=1$, we  use Lemma~\ref{cor-location-scale-transformation} to obtain $f(\omega)/f(y_i)=f((y_i-a_i)/b_i)/f(y_i) \le |b_i|D(|a_i|,|b_i|)$.
  In step $b$, we use $\pi(\beta,\sigma)\le \max(\sigma^{-1},1)B= \sigma^{-1} \max(1,\sigma)B=\sigma^{-1}B$.
 For the case $k_i=1$, we use the monotonicity of the tails of $f$ to obtain
\begin{align*}
|x_i|^{-\theta}f\left(\sigma^{-1}|x_i|^{-\theta}(a_i-\beta x_i)\right)&\le \zeta^{|\theta|} f(\sigma^{-1}\zeta^{-(2+|\theta|)}\omega/4) \cr
&\le 4\zeta^{2|\theta|+2} D(0,4\zeta^{2+|\theta|}) f(\omega/\sigma)
\end{align*}
 because, if we define the constant $a_{(k)}:=\max_{\{i\,:\,k_i=1\}}|a_i|$ with $\omega\ge \yo\ge 4\zeta^2 a_{(k)}$, we have
$\sigma^{-1}|x_i|^{-\theta}|a_i-\beta x_i|\geq \sigma^{-1}|x_i|^{-\theta}(\beta\abs{x_i}-\abs{a_i})\ge \sigma^{-1}\zeta^{-|\theta|}$ $\times((\zeta^{-1}\omega/2) \zeta^{-1}-a_{(k)}) \ge \sigma^{-1}\zeta^{-|\theta|}(\zeta^{-2}\omega/2-\zeta^{-2}\omega/4)=\sigma^{-1}\zeta^{-(2+|\theta|)}$ $\times\omega/4
 \ge\zeta^{-(2+|\theta|)}\omega/4 \ge \zeta^{-(2+|\theta|)}\yo/4\ge M$ if we choose $\yo\ge 4\zeta^{2+|\theta|} M$. We use Lemma~\ref{cor-location-scale-transformation} in the second inequality.
 In step $c$, we use the monotonicity of the tails of $|z|f(z)$ to obtain $\sigma^{-1}\omega f(\sigma^{-1}\omega)\le \omega f(\omega)$ because $\sigma^{-1}\omega\ge \omega \ge \yo\ge M$ if we choose $\yo\ge M$.
In step $d$, we assume for convenience and without loss of generality that $\{i:p_i=1\}=\{1,\ldots,p\}$, and we
 consider this assumption for the rest of the proof.

As in the proof of Proposition~1,
 we now split the real line (which includes $\zeta^{-1}\omega/2\le \beta <\infty$) into $p$ mutually disjoint intervals given by
 $(y_{j-1}/x_{j-1}+y_{j}/x_{j})/2\le \beta \le (y_{j}/x_{j}+y_{j+1}/x_{j+1})/2,$ for $j=1,\ldots,p$,
 where we define $y_0/x_0:=-\infty$ and $y_{p+1}/x_{p+1}:=\infty$.
  We also define the constant $\delta>0$ as follows:
  \begin{equation*}
  \delta=\zeta^{-1}\times\min_{i\in\{1,\ldots,p-1\}}\left\{(y_{i+1}/x_{i+1}-y_i/x_i)/2\right\}.
  \end{equation*}

Consider   $(y_{j-1}/x_{j-1}+y_{j}/x_{j})/2\le \beta \le (y_{j}/x_{j}+y_{j+1}/x_{j+1})/2$, for $j\in\{1,\ldots,p\}$, and $0<\sigma< 1$. Thus,
 \begin{align*}
 &\sigma^{-1}\left[\sigma^{-1}f(\sigma^{-1}\omega)\right]^{k-p}\prod_{i=1}^{p}\sigma^{-1}|x_i|^{-\theta}f\left(\sigma^{-1}|x_i|^{-\theta}(y_i-\beta x_i)\right)\\
&\za{\le}(\delta^{-1} B)^{p-1} \sigma^{-1}\left[\sigma^{-1}f(\sigma^{-1}\omega)\right]^{k-p}\sigma^{-1}|x_j|^{-\theta}f\left(\sigma^{-1}|x_j|^{-\theta}(y_j-\beta x_j)\right)\\
&\zb{\le} (\delta^{-1} B)^{p-1}B^{k-p-1}\omega^{-(k-p)}\sigma^{-2}\omega f(\sigma^{-1}\omega)\times\frac{1}{\sigma|x_j|^{\theta}}f\left(\frac{y_j-\beta x_j}{\sigma|x_j|^{\theta}}\right).
\end{align*}
In step $a$, we use, for $i\ne j$, $\sigma^{-1}|x_i|^{-\theta}f\left(\sigma^{-1}|x_i|^{-\theta}(y_i-\beta x_i)\right)\le |y_i-\beta x_i|^{-1}B\le \delta^{-1} B$, where we bound $|z|f(z)$ by $B$ and we
 use $|y_i-\beta x_i|\ge\delta$ since
   \begin{align*}
   &|y_i-\beta x_i|=|x_i||y_i/x_i-\beta|\ge \zeta^{-1}|y_i/x_i-\beta|\\
   &\hspace{1.5cm}\ge \zeta^{-1}\times\min\left\{(y_{j}/x_{j}-y_{j-1}/x_{j-1})/2,(y_{j+1}/x_{j+1}-y_{j}/x_{j})/2\right\}\ge \delta.
\end{align*}
In step $b$, we use $\sigma^{-1}\omega f(\sigma^{-1}\omega)\le B$ for $k-p-1$ terms (by assumption $k>\max(m,p)\Rightarrow k>p$).

 Finally, we have
\begin{align*}
 &\omega^{-(k-p)}\int_{0}^{1} \sigma^{-2}\omega f(\sigma^{-1}\omega) \int_{(y_{j-1}/x_{j-1}+y_{j}/x_{j})/2}^{(y_{j}/x_{j}+y_{j+1}/x_{j+1})/2}\frac{1}{\sigma|x_j|^{\theta}}f\left(\frac{y_j-\beta x_j}{\sigma|x_j|^{\theta}}\right)\,d\beta\,d\sigma\\
 &\le \omega^{-(k-p)}\int_{0}^{\infty}\sigma^{-2}\omega f(\sigma^{-1}\omega)  \int_{-\infty}^{\infty}\sigma^{-1}|x_j|^{-\theta}f\left(\sigma^{-1}|x_j|^{-\theta}(y_j-\beta x_j)\right)\,d\beta\,d\sigma\\
& \za{=} |x_j|^{-1}\omega^{-(k-p)}\int_{0}^{\infty}f(\sigma')\,d\sigma'\int_{-\infty}^{\infty}f(\beta')\,d\beta'
\le \zeta \,\omega^{-(k-p)}\zb{\rightarrow} 0 \text{ as }\omega\rightarrow\infty.
\end{align*}
In step $a$, we use the change of variables  $\sigma'=\sigma^{-1}\omega$ and $\beta'=
\sigma^{-1}|x_j|^{-\theta}(y_j-\beta x_j)$. In step $b$, we use $k>p$.
\end{proof}

\begin{proof}[Proof of Result~(b)]
Consider $(\beta,\sigma)$ such that $\pi(\beta,\sigma)>0$ (the proof for the case $(\beta,\sigma)$ such that $\pi(\beta,\sigma)=0$ is trivial).
We have, as $\omega\rightarrow \infty$,
  \begin{align*}
\frac{\pi(\beta,\sigma\mid \mathbf{y_n})}{\pi(\beta,\sigma\mid \mathbf{y_k})}
&= \frac{m(\mathbf{y_k})}{m(\mathbf{y_n})}\times
\frac{\pi(\beta,\sigma)\prod_{i=1}^{n}(\sigma|x_i|^{\theta})^{-1}f\left((\sigma |x_i|^{\theta})^{-1}(y_i-\beta x_i)\right)}
{\pi(\beta,\sigma)\prod_{i=1}^{n}\left[(\sigma |x_i|^{\theta})^{-1}f\left((\sigma |x_i|^{\theta})^{-1}(y_i-\beta x_i)\right)\right]^{k_i}}\\
&= \frac{m(\mathbf{y_k})}{m(\mathbf{y_n})}\prod_{i=1}^{n}\left[(\sigma |x_i|^{\theta})^{-1}f\left((\sigma |x_i|^{\theta})^{-1}(y_i-\beta x_i)\right)\right]^{m_i+p_i}\\
&\hspace{-0mm}= \frac{m(\mathbf{y_k})\prod_{i=1}^{n}[f(y_i)]^{m_i+p_i}}{m(\mathbf{y_n})}\prod_{i=1}^{n}\left[\frac{\frac{1}{\sigma |x_i|^{\theta}}f\left((\frac{y_i-\beta x_i}{\sigma |x_i|^{\theta}}\right)}
{f(y_i)}\right]^{m_i+p_i}\rightarrow 1.
   \end{align*}

The first ratio in the last equality does not depend on $\beta$ and $\sigma$, and converges towards 1 as $\omega\rightarrow\infty$ using result~(a). The second part also converges to 1 uniformly
in any set $(\beta,\sigma)\in [-\lambda,\lambda]\times[1/\tau, \tau]$ using Proposition~4 of \cite{desgagne2015robustness} since $x_1,\ldots, x_n$ and $\theta$ are fixed.
Furthermore, since $f$ and $\sigma\pi(\beta,\sigma)$ are bounded, and $x_i\ne 0$ for all $i$, $\pi(\beta,\sigma\mid \mathbf{y_k})$ is also bounded on any set $(\beta,\sigma)\in [-\lambda,\lambda]\times[1/\tau, \tau]$.
 Then, we have
   \begin{equation*}
   \big|\pi(\beta,\sigma\mid \mathbf{y_n})-\pi(\beta,\sigma\mid \mathbf{y_k})\big|=\pi(\beta,\sigma\mid \mathbf{y_k})\abs{\frac{\pi(\beta,\sigma\mid \mathbf{y_n})}{\pi(\beta,\sigma\mid \mathbf{y_k})}-1}\rightarrow 0\text{ as }\omega\rightarrow\infty.
   \end{equation*}
\end{proof}

\begin{proof}[Proof of Results (c) and (d)]

Using Proposition~\ref{proposition-proper}, we know that $\pi(\beta,\sigma\mid\mathbf{y_k})$ and $\pi(\beta,\sigma\mid\mathbf{y_n})$ are proper. Moreover, using result~(b), we have the pointwise convergence $\pi(\beta,\sigma\mid \mathbf{y_n})\rightarrow\pi(\beta,\sigma\mid \mathbf{y_k})$   as $\omega\rightarrow\infty$ for any $\beta\in\re$ and $\sigma>0$, as a result of the uniform convergence. Then, the conditions of Scheff\'{e}'s theorem (see \cite{scheffe1947useful}) are satisfied and we obtain the convergence in $L_1$ given by result~(c) as well as the following result:
 \begin{equation*}
 \lim_{\omega\rightarrow\infty}\int_{E}\pi(\beta,\sigma\mid\mathbf{y_n})\,d\beta\,d\sigma=
  \int_{E}\pi(\beta,\sigma\mid\mathbf{y_k})\,d\beta\,d\sigma,
  \end{equation*}
uniformly for all rectangles $E$ in $\re\times \re^{+}$. Result~(d) follows directly.
%
%
\end{proof}
\begin{proof}[Proof of Result~(e)]
  Using equation~(\ref{likeli1}), result~(e) follows directly from result~(b). 
%
%
\end{proof}

\subsection{R Functions}\label{r_functions}

In this section, we provide the R functions that were used for the computations. They follow the same order as the numerical results presented in the paper. We start with the computer code needed to produce Figure~\ref{fig_examples}, and it is followed by that used for the numerical results contained in Sections~\ref{sec_illu_thm} and \ref{sec_comparison} of the paper.

%

For the \textit{M}-estimator in Section~\ref{sec_comparison}, we use the \textit{rlm} function from the \textit{MASS} package (\cite{masspackage}). The \textit{lmrob.S} function from the \textit{robustbase} package (\cite{robustbase}) is used for the \textit{S}-estimator. This last function was built according to the fast algorithm of \cite{salibian2006fast}.

\begin{verbatim}
###################### General Functions #####################

# Note that these R functions can also be used for a
# location-scale model (beta and sigma) if we replace the
# x observations by x <- rep(1, length(y))

normal_model <- function(x, y, theta){
  w <- abs(x) ^ (2 * (1 - theta))
  w <- w / sum(w)
  beta_estimate <- sum(w * y / x)
  sigma_estimate <- sqrt(mean((y - beta_estimate * x ) ^ 2 /
                    (abs(x)) ^ (2 * theta)))
  return(c(beta_estimate, sigma_estimate))
}

student_model <- function(x, y, theta, df, beta0, sigma0){
  # negative of the log-likelihood (up to a constant)
  neg_log_likely_student <- function(param, x, y, theta, df){
    # known_scale added to match 2.5 & 97.5th perc. of N(0,1)
    known_scale <- round(qnorm(0.025) / qt(.025, df), 2)
    z <- abs((y - param[1] * x) / (param[2] * abs(x) ^ theta))
    logL <- sum(dt(z / known_scale, df = df, log = TRUE) -
            log(param[2]))
    return( - logL)
  }
  estimates_student <- optim(c(beta0, sigma0),
    neg_log_likely_student, gr = NULL, x = x, y = y,
    theta = theta, df = df, method = "L-BFGS-B",
    lower = c(-Inf, 1e-08), upper = c(Inf, Inf),
    control = list(factr = 5))$par
  return(estimates_student)
}

LPTN_model <- function(x, y, theta, alpha, beta0, sigma0){
  # negative of the log-likelihood (up to a constant)
  neg_log_likely_LPTN <- function(param, x, y, theta, alpha){
    q <- 2 * pnorm(alpha) - 1
    phi <- 1 + 2 * dnorm(alpha) * alpha * log(alpha) / (1 - q)
    if (param[2] <= 0){return(Inf)} else{
    z <- abs((y - param[1] * x) / (param[2] * abs(x) ^ theta))
    tails <- as.numeric(z <= alpha)
    # to avoid undefined value of log(log(z)) if tails = 1
    z_floor <- apply(cbind(z), 1, max, 1.001)
    logL <- sum(tails * dnorm(z, log = TRUE) + (1 - tails) *
      (dnorm(alpha, log = TRUE) + log(alpha) - log(z_floor) +
       phi * log(log(alpha)) - phi * log(log(z_floor))) -
       log(param[2]))
    return( - logL)
  }}
  estimates_LPTN <- optim(c(beta0, sigma0),
      neg_log_likely_LPTN, gr = NULL, x = x, y = y,
      theta = theta, alpha = alpha, method = "Nelder-Mead",
      control = list(maxit = 40000, reltol=10^(-12) ))$par
  return(estimates_LPTN)
}

dLPTN <- function(y, alpha, mu = 0, sigma = 1, log.d = FALSE){
  # alpha must be larger than 1
  # the density is a N(mu, sigma^2) between
  # mu - alpha * sigma and mu + alpha * sigma, with a
  # mass of q, the density has log-Pareto tails propto
  # (1 / |z|) * (log|z|) ^ ( - phi)
  q <- 2 * pnorm(alpha) - 1
  phi <- 1 + 2 * dnorm(alpha) * alpha * log(alpha) / (1 - q)
  z <- abs((y - mu) / sigma)
  tails <- as.numeric(z <= alpha)
  # to avoid undefined value of log(log(z)) if tails = 1 :
  z_floor <- apply(cbind(z), 1, max, 1.001)
  # or equivalently z_floor <- z + 2 * tails
  logf <- tails * dnorm(z, log = TRUE) + (1 - tails) *
          (dnorm(alpha, log = TRUE) + log(alpha) -
           log(z_floor) + phi * log(log(alpha)) -
           phi * log(log(z_floor))) - log(sigma)
   if (log.d == TRUE) {res <- logf} else {res <- exp(logf)}
   return(res)
}

######################## Figure 1 #########################

# Total cultivated area in 1931 (acres)
x <- c(401, 634, 1194, 1770, 1060, 827, 1737, 1060, 360,
       946, 470, 1625, 827, 96, 1304, 377, 259, 186, 1767,
       604, 701, 524, 571, 962, 407, 715, 845, 1016, 184,
       282, 194, 439, 854, 824)

# Area under wheat in 1936 (acres)
y <- c(75, 163, 326, 442, 254, 125, 559, 254, 101, 359, 109,
       481, 125, 5, 427, 78, 78, 45, 564, 238, 92, 247, 134,
       131, 129, 192, 663, 236, 73, 62, 71, 137, 196, 255)

par(mar = c(4.5, 5, 1, 3))
plot(x, y, xlab = "Total cultivated area in 1931 (acres)",
     ylab = "Area under wheat 1936 (acres)", pch = 19,
     cex.lab = 1.5, cex.axis = 1.5, cex = 1)

beta_normal <- normal_model(x, y, theta = 0.5)[1]
beta_LPTN <- LPTN_model(x, y, theta = 0.5, alpha = 1.96,
                        beta0 = 0.27, sigma0 = 2)[1]

abline(a = 0, b = beta_normal, col = "darkorange", lwd = 3,
       lty = 5)
abline(a = 0, b = beta_LPTN, col = "darkblue", lwd = 3)

####################### Section 3.1 #########################

# Number of occupied dwellings in 1960
x1 <- c(82, 61, 42, 51, 58, 50, 60, 50, 54, 50, 51, 54, 27,
        25, 48, 50, 38, 43, 48, 50, 48, 70, 13, 56)
x2 <- c(50, 11, 31, 29, 45, 40, 43, 5, 40, 37, 48, 46, 55,
        45, 43, 51, 48, 49, 51, 41, 45, 42, 51, 48, 42, 58,
        63, 51)
x3 <- c(48, 53, 48, 31, 46, 43, 51, 42, 52, 57, 49, 50, 51,
        64, 76, 71, 44, 41, 39, 44, 43, 47, 49, 10, 10, 36,
        31, 41)
x4 <- c(36, 47, 27, 17, 21, 9, 12, 33, 21, 22, 30, 35, 46, 17,
        18, 18, 19, 50, 36, 60, 56, 46, 42, 36, 47, 34, 21)
x5 <- c(26, 12, 41, 20, 57, 62, 24, 20, 30, 25, 31, 10, 15,
        16, 37, 45, 63, 53, 34, 61, 52, 80, 14, 50, 25, 34,
        38, 63)
x6 <- c(73, 51, 47, 53, 68, 83, 113, 55, 64, 45, 46, 35, 26,
        45, 68, 32, 53, 58, 49, 45, 70, 64, 97, 103, 85, 60,
        57)
x7 <- c(39, 36, 41, 41, 43, 31, 77, 70, 43, 46, 53, 73, 63,
        46, 59, 54, 64, 58, 76, 32, 18, 49, 56, 41, 48, 33,
        55, 24)
x8 <- c(51, 48, 36, 54, 46, 57, 46, 37, 54, 61, 56, 56, 50, 21,
        19, 15, 8, 37, 42, 74, 81, 28, 17, 55, 58, 20, 46, 33)
x9 <- c(48, 50, 47, 85, 53, 50, 64, 52, 56, 22, 22, 17, 12, 25,
        33, 69, 34, 17, 16, 66, 38, 24, 12, 54, 48, 53, 54, 12)
x10 <- c(107, 134, 130, 72, 56, 46, 41, 18, 36, 38, 35, 23, 28,
         43, 27, 51, 16, 29, 30, 47, 18, 9, 26, 62)
x11 <- c(34, 7, 270, 169, 84, 146, 8, 6, 27, 35, 12, 22, 29,
         29, 32, 44, 59, 65, 73, 71, 74, 62, 111, 124, 28)
x12 <- c(38, 63, 15, 41, 68, 57, 74, 51, 64, 44, 39, 21, 49,
         64, 49, 84, 66, 64, 73, 54, 49, 55, 47, 44, 62, 49,
         56, 43, 53, 58)
x13 <- c(63, 50, 38, 44, 61, 66, 62, 50, 64, 57, 59, 62, 62,
         55, 60, 47, 51, 56, 66, 34, 65, 26, 56, 53, 53, 41,
         36)
x14 <- c(32, 14, 5, 12, 51, 57, 66, 65, 67, 62, 67, 68, 52, 40,
         50, 57, 51, 53, 57, 46, 51, 50, 44, 48, 52, 46, 19,
         49)
x15 <- c(42, 45, 46, 52, 54, 50, 52, 9, 51, 49, 55, 30, 53, 44,
         50, 46, 37, 43, 25, 45, 60, 52, 53, 60, 61, 64, 59,
         58)
x16 <- c(56, 60, 98, 93, 50, 52, 40, 45, 26, 16, 45, 54, 49,
         47, 49, 62, 52, 85, 24, 63, 51, 38, 51, 5, 50, 92, 56)
x17 <- c(64, 68, 28, 13, 40, 18, 44, 37, 32, 20, 32, 14, 28,
         47, 36, 44, 73, 81, 14, 50, 11, 41, 24, 23, 35, 23)
x18 <- c(9, 41, 26, 11, 12, 16, 26, 45, 59, 36, 23, 26, 30, 53,
         45, 58, 56, 59, 31, 20, 56, 58, 57, 51, 30, 25, 31,
         46, 28)
x19 <- c(16, 29, 26, 45, 20, 38, 37, 35, 23, 23, 31, 41, 39,
         30, 27, 22, 37, 43, 46, 46, 29, 31, 13, 30, 48, 35)
x20 <- c(40, 46, 27, 24, 32, 45, 47, 45, 33, 44, 32, 35, 12,
         47, 43, 49, 13, 19, 42, 42, 39, 45, 51, 50, 32, 63,
         68, 43)
x21 <- c(44, 52, 111, 67, 57, 54, 57, 134, 62, 70, 41, 37, 16)
x <- c(x1, x2, x3, x4, x5, x6, x7, x8, x9, x10, x11, x12, x13,
       x14, x15, x16, x17, x18, x19, x20, x21)

# Number of persons in 1970
y1 <- c(185, 145, 127, 136, 122, 116, 165, 134, 174, 141, 151,
        138, 90, 78, 129, 139, 72, 127, 153, 120, 132, 202,
        140, 375)
y2 <- c(136, 42, 19, 103, 151, 166, 177, 31, 156, 125, 172,
        141, 146, 49, 107, 177, 149, 134, 126, 119, 168, 95,
        96, 134, 102, 118, 147, 129)
y3 <- c(140, 130, 119, 89, 104, 103, 126, 116, 127, 151, 178,
        131, 132, 241, 185, 315, 158, 181, 92, 104, 135, 167,
        123, 42, 39, 77, 66, 88)
y4 <- c(110, 129, 117, 50, 41, 27, 14, 14, 33, 107, 76, 119,
        135, 62, 39, 39, 55, 150, 145, 214, 171, 198, 169, 128,
        149, 120, 42)
y5 <- c(117, 36, 102, 68, 186, 195, 64, 55, 77, 91, 115, 25,
        32, 39, 99, 77, 124, 165, 109, 143, 205, 157, 45, 92,
        74, 74, 66, 153)
y6 <- c(135, 59, 185, 138, 169, 211, 215, 164, 171, 164, 171,
        80, 52, 58, 118, 112, 136, 180, 156, 161, 188, 164,
        255, 253, 177, 174, 132)
y7 <- c(118, 68, 86, 152, 83, 85, 151, 212, 143, 159, 116,
        156, 125, 89, 150, 165, 303, 198, 237, 178, 60, 174,
        130, 126, 128, 100, 152, 75)
y8 <- c(108, 89, 151, 146, 0, 0, 70, 131, 89, 195, 0, 0, 79,
        30, 27, 26, 4, 82, 100, 211, 161, 91, 135, 112, 160,
        65, 187, 87)
y9 <- c(240, 144, 64, 329, 232, 32, 90, 69, 147, 0, 0, 14, 35,
        1, 18, 65, 88, 19, 15, 66, 123, 0, 0, 47, 82, 137,
        159, 698)
y10 <- c(72, 618, 256, 0, 0, 116, 32, 2, 53, 97, 11, 37, 73,
         104, 86, 136, 30, 0, 44, 73, 39, 3, 63, 0)
y11 <- c(242, 51, 781, 358, 405, 406, 53, 98, 113, 91, 27, 74,
         58, 41, 17, 72, 221, 236, 153, 227, 119, 138, 256,
         257, 37)
y12 <- c(86, 131, 16, 101, 83, 120, 116, 161, 179, 69, 130,
         112, 88, 229, 151, 247, 193, 325, 550, 162, 161, 155,
         166, 133, 175, 148, 155, 161, 142, 176)
y13 <- c(182, 182, 91, 171, 176, 176, 161, 185, 196, 188, 174,
         163, 166, 159, 168, 135, 139, 157, 133, 132, 270, 80,
         145, 158, 162, 106, 69)
y14 <- c(90, 43, 0, 57, 201, 237, 89, 274, 177, 213, 167, 116,
         107, 50, 87, 104, 156, 96, 97, 104, 150, 131, 109,
         142, 151, 103, 39, 148)
y15 <- c(95, 120, 115, 111, 132, 111, 124, 26, 96, 123, 134,
         87, 146, 105, 117, 133, 123, 155, 15, 153, 139, 229,
         139, 155, 180, 148, 168, 165)
y16 <- c(145, 345, 135, 239, 117, 193, 137, 106, 100, 37, 112,
         120, 154, 121, 146, 140, 138, 149, 138, 150, 118,
         107, 121, 2, 70, 236, 178)
y17 <- c(188, 162, 32, 94, 166, 42, 82, 81, 118, 37, 122, 133,
         92, 81, 165, 139, 124, 174, 68, 176, 27, 111, 47, 19,
         108, 21)
y18 <- c(8, 60, 12, 28, 29, 30, 79, 117, 100, 127, 95, 119, 92,
         142, 118, 156, 154, 177, 71, 72, 209, 187, 154, 147,
         52, 105, 76, 65, 82)
y19 <- c(8, 62, 29, 136, 126, 92, 112, 52, 42, 59, 58, 146, 55,
         177, 68, 74, 45, 141, 190, 96, 77, 86, 43, 56, 99,
         143)
y20 <- c(108, 151, 79, 62, 102, 117, 165, 115, 89, 112, 93, 44,
         69, 107, 100, 133, 17, 36, 82, 84, 76, 101, 112, 124,
         72, 226, 206, 63)
y21 <- c(130, 336, 445, 116, 133, 138, 168, 226, 87, 278, 109,
         121, 131)
y <- c(y1, y2, y3, y4, y5, y6, y7, y8, y9, y10, y11, y12, y13,
       y14, y15, y16, y17, y18, y19, y20, y21)

par(mar = c(4.5, 5, 1, 3))
plot(x, y, xlab = "Number of occupied dwellings in 1960",
     ylab = "Number of persons in 1970", pch = 19,
     cex.lab = 1.5, cex.axis = 1.5, cex = 1,
     xlim = c(0, 300), ylim = c(0, 850))

beta_normal <- normal_model(x, y, theta = 0.5)[1]
beta_LPTN <- LPTN_model(x, y, theta = 0.5, alpha = 1.96,
                        beta0 = 2.60, sigma0 = 6.00)[1]

abline(a = 0, b = beta_normal, col = "darkorange", lwd = 3,
       lty = 5)
abline(a = 0, b = beta_LPTN, col = "darkblue", lwd = 3)

###################### Table 1: Data set #####################

x <- c(1.0, 1.0, 2.0, 3.0, 3.0, 3.0, 3.0, 3.0, 3.0, 3.0,
       3.0, 4.0, 4.0, 4.0, 4.0, 5.0, 5.0, 5.0, 6.0, 6.0)

y <- c(20.8, 9.6, 38.6, 74.1, 108.8, 98.7, 44.8, 77.2, 93.2,
       107.2, NA, 93.6, 113.7, 123.5, 93.5, 148.1, 147.1,
       154.0, 149.5, 173.5)

# The computations for Figure 2

y11 <- seq(100, 385, 0.1)
n11 <- length(y11)
estimates_normal <- estimates_student <- estimates_LPTN <-
  matrix(ncol = 2, nrow = n11)

for (i in 1:n11) {
  y[11] <- y11[i]
  estimates_normal[i,] <- normal_model(x, y, theta = 0.5)
  estimates_student[i,] <- student_model(x, y, theta = 0.5,
                           df = 10, beta0 = 28, sigma0 = 14)
  estimates_LPTN[i,] <- LPTN_model(x, y, theta = 0.5,
                        alpha = 1.96, beta0 = 28, sigma0 = 14)
}

########################## Figure 2a #########################

par(mar = c(4.5, 6, 1, 3))
plot(y11, estimates_normal[,1], type = "l", xlim = c(100,
     385), ylim = c(27.4, 32), col = "darkorange",
     cex.lab = 1.5, cex.axis = 1.5, cex = 1.5, lwd = 3,
     lty = 5, xlab = expression(y[11]),
     ylab = expression(hat(beta)))
lines(y11, estimates_student[,1], type = "l",
     col = "darkgreen", lwd = 3, lty = 4)
lines(y11, estimates_LPTN[,1], type = "l", col = "darkblue",
     lwd = 3)

########################## Figure 2b #########################

par(mar = c(4.5, 6, 1, 3))
plot(y11, estimates_normal[,2], type = "l", xlim = c(100,
     385), ylim = c(10, 30), col = "darkorange",
     cex.lab = 1.5, cex.axis = 1.5, cex = 1.5, lwd = 3,
     lty = 5, xlab = expression(y[11]),
     ylab = expression(hat(sigma)))
lines(y11, estimates_student[,2], type = "l",
     col = "darkgreen", lwd = 3, lty = 4)
lines(y11, estimates_LPTN[,2], type = "l", col = "darkblue",
     lwd = 3)

########################## Figure 3a #########################

z <- seq(-4, 4, .01)
par(mar = c(4.5, 5, 1, 3))
plot(z, dLPTN(z, alpha = 1.96), type = "l", col = "darkblue",
     cex.lab = 1.5, cex.axis = 1.5, cex = 1.5, lwd = 3,
     xlab = expression(x), ylab = expression(f(x)))
lines(z, dnorm(z), type = "l", col = "darkorange", lwd = 3,
     lty = 5)

########################## Figure 3b #########################

z <- seq(1.96, 6, .01)
par(mar = c(4.5, 5, 1, 3))
plot(z, dLPTN(z, alpha = 1.96), type = "l", col = "darkblue",
     cex.lab = 1.5, cex.axis = 1.5, cex = 1.5, lwd = 3,
     xlab = expression(x), ylab = expression(f(x)))
lines(z, dnorm(z), type = "l", col = "darkorange", lwd = 3,
     lty = 5)
	
################ Location of the threshold ###################
	
beta_LPTN_maxinfluence <- max(estimates_LPTN[,1])
y11_LPTN_maxinfluence <- y11[which(estimates_LPTN[,1] ==
    beta_LPTN_maxinfluence)]
print(c(y11_LPTN_maxinfluence, beta_LPTN_maxinfluence ))
[1] 127.9000  28.6259

sigma_LPTN_maxinfluence <- max(estimates_LPTN[,2])
y11_LPTN_maxinfluence <- y11[which(estimates_LPTN[,2] ==
    sigma_LPTN_maxinfluence)]
print(c(y11_LPTN_maxinfluence, sigma_LPTN_maxinfluence))
[1] 127.90000  12.37836

################### y[11] goes to infinity ###################

y[11] <- 10^155
LPTN_model(x, y, 0.5, 1.96, 28, 11)
[1] 27.13129 10.77932

############### Inference without x[11], y[11] ###############

LPTN_model(x[-11], y[-11], 0.5, 1.96, 28, 11)
27.13016 10.77833

###################### Table 2: Data set #####################

x <- c(102.9, 144.9, 155.8, 176.5, 177.4, 182.2, 197.9, 199.2,
       211.3, 215.9, 216.0, 216.7, 220.3, 222.8, 229.0, 250.0,
       250.2, 275.4, 342.4, 696.4)

y <- c(31.7, 68.4, 54.4, 53.5, 78.4, 66.4, 64.1, 44.6, 99.0,
       53.3, 67.3, 68.6, 63.0, 100.6, 82.2, 113.4, 6.1, 76.6,
       92.7, 41.1)

############### First computation method: MCMC ###############
# A random walk Metropolis algorithm is used.
# The results in the paper were produced using this method.

library(PoweR)
library("coda")

# log-likelihood function under the normal (up to a constant)
log_likely_normal <- function(param, x, y, theta){
    z <- abs((y - param[1] * x) / (param[2] * abs(x) ^ theta))
    logL <- sum(dnorm(z, log = TRUE) - log(param[2]))
    return(logL)
}

# log-likelihood function under the Student (up to a constant)
log_likely_student <- function(param, x, y, theta){
    df <- 10
    # known_scale added to match 2.5 & 97.5th perc. of N(0,1)
    known_scale <- round(qnorm(0.025) / qt(.025, df), 2)
    z <- abs((y - param[1] * x) / (param[2] * abs(x) ^ theta))
    logL <- sum(dt(z / known_scale, df = df, log = TRUE) -
            log(param[2]))
    return(logL)
}

# log-likelihood function under the LPTN (up to a constant)
log_likely_LPTN <- function(param, x, y, theta){
    alpha <- 1.96
    q <- 2 * pnorm(alpha) - 1
    phi <- 1 + 2 * dnorm(alpha) * alpha * log(alpha) / (1 - q)
    if (param[2] <= 0){return(Inf)} else{
    z <- abs((y - param[1] * x) / (param[2] * abs(x) ^ theta))
    tails <- as.numeric(z <= alpha)
    # to avoid undefined value of log(log(z)) if tails = 1
    z_floor <- apply(cbind(z), 1, max, 1.001)
    logL <- sum(tails * dnorm(z, log = TRUE) + (1 - tails) *
      (dnorm(alpha, log = TRUE) + log(alpha) - log(z_floor) +
       phi * log(log(alpha)) - phi * log(log(z_floor))) -
       log(param[2]))
    return(logL)
}}

rLPTN <- function(n){
  alpha <- 1.96
  return(gensample(40, n, law.pars = c(alpha, 0.0, 1.0),
           check = FALSE)$sample)
}

rt10 <- function(n){
  df = 10
  # known_scale added to match 2.5 & 97.5th perc. of N(0,1)
  known_scale <- round(qnorm(0.025) / qt(.025, df), 2)
  return(known_scale * rt(n, df))
}

mcmc <- function(nb_iter, x, y, theta, initial_val = c(0, 1),
           scaling_prop = 0.10, law = rnorm, log_likely
           = log_likely_normal){
    # initial_val: we use the posterior medians after
    # some trial runs
    # scaling_prop = scaling for the random walk
    # it is efficient to use 0.15 with the outliers
    # and 0.10 without the outliers, for all models
    # law for the proposals = rnorm or rt10 or rLPTN
    nb_accept <- 0	
    n <- length(y)
    matrix_res <- matrix(ncol = 2, nrow = nb_iter + 1)
    matrix_res[1,] <- initial_val
    for(i in 2:(nb_iter+1)){
      # location = current state
      location <- as.matrix(matrix_res[(i-1),])
      ### generate the candidate
      w <- location + scaling_prop * law(2)
      # compute the acceptance probability	
      if(w[2] > 0){ # check that candidate for sigma > 0			
        # log of numerator, note: prior is 1 / w[2]
        log_num <- - log(w[2]) + log_likely(w, x, y, theta)			
        # log of denominator, note: prior is 1 / w[2]
        log_denom <- - log(location[2]) +
                          log_likely(location, x, y, theta)			
        if(log(runif(1)) <= log_num - log_denom){
          # accept the candidate
                 matrix_res[i,] <- w
                 nb_accept <- nb_accept + 1
        } else{
                 matrix_res[i,] <- matrix_res[(i-1),]				
        }
      } else{		
            matrix_res[i,] <- matrix_res[(i-1),]	
      }
    }
    return(list(estim = matrix_res[2:(nb_iter+1),],
           rate = nb_accept/nb_iter))
}

set.seed(1)
sim_norm <- mcmc(nb_iter = 10 ^ 7, x, y, theta = 0.5,
            initial_val = c(0.28, 2.2), scaling_prop = 0.15,
            law = rnorm, log_likely = log_likely_normal)

set.seed(1)
sim_student <- mcmc(nb_iter = 10 ^ 7, x, y, theta = 0.5,
            initial_val = c(0.3, 2), scaling_prop = 0.15,
            law = rt10, log_likely = log_likely_student)

set.seed(1)
sim_LPTN <- mcmc(nb_iter = 10 ^ 7, x, y, theta = 0.5,
            initial_val = c(0.32, 1.6), scaling_prop = 0.15,
            law = rLPTN, log_likely = log_likely_LPTN)

print(c(sim_norm$rate, sim_student$rate, sim_LPTN$rate))
[1] 0.2449020 0.2539403 0.2316780

beta_post <- cbind(sim_norm$estim[,1], sim_student$estim[,1],
             sim_LPTN$estim[,1])
beta_median <- apply(beta_post, 2, median)
print(beta_median)
[1] 0.2829841 0.3061633 0.3185546

sigma_post <- cbind(sim_norm$estim[,2], sim_student$estim[,2],
              sim_LPTN$estim[,2])
sigma_median <- apply(sigma_post, 2, median)
print(sigma_median)
[1] 2.180424 2.031161 1.634240

######################## HPD intervals #######################

HPD_beta <- apply(beta_post, 2, function(gen){HPDinterval(
            as.mcmc(gen), prob = 0.95)[1:2]})
print(HPD_beta)
[1,] 0.2168912 0.2427309 0.2399747
[2,] 0.3480296 0.3666505 0.3762819

HPD_sigma <- apply(sigma_post, 2, function(gen){HPDinterval(
            as.mcmc(gen), prob = 0.95)[1:2]})
print(HPD_sigma)
[1,] 1.565228 1.319387 0.9618679
[2,] 3.015768 2.958068 2.6742025

########################## Figure 4a #########################

par(mar = c(4.5, 5, 1, 3))
plot(x, y, type = "p", pch = 19, cex.lab = 1.5,
       cex.axis = 1.5, cex = 1, xlab = "Income",
       ylab = "Expenditure on food", xlim = c(0, 700),
       ylim = c(0, 120))
abline(a = 0, b = beta_median[1], col = "darkorange", lwd = 3,
       lty = 5)
abline(a = 0, b = beta_median[2], col = "darkgreen", lwd = 3,
       lty = 4)
abline(a = 0, b = beta_median[3], col = "darkblue", lwd = 3)

########################## Figure 4b #########################

par(mar = c(4.5, 6, 1, 3))
plot(density(beta_post[,1], adjust=2), xlim = c(0.15, 0.45),
       ylim = c(0, 14), col = "darkorange", cex.lab = 1.5,
       cex.axis = 1.5, cex = 1, lwd = 3, lty = 5,
       xlab = expression(beta), main = "",
       ylab = expression(pi(beta ~"|"~ bold(y[n]))))	
lines(c(0.2168912, 0.3480296), c(1.5, 1.5), col = "darkorange",
       lwd = 3, lty = 5)
lines(density(beta_post[,2], adjust=2), col = "darkgreen",
       lwd = 3, lty = 4)
lines(c(0.2427309, 0.3666505), c(1.7, 1.7), col = "darkgreen",
       lwd = 3, lty = 4)
lines(density(beta_post[,3], adjust=2), col = "darkblue",
       lwd = 3)
lines(c(0.2399747, 0.3762819), c(1.6, 1.6), col = "darkblue",
       lwd = 3)

########################## Figure 4c #########################

par(mar = c(4.5, 6, 1, 3))
plot(density(sigma_post[,1], adjust=2), xlim = c(0.5, 4.5),
       ylim = c(0, 1.3), col = "darkorange", cex.lab = 1.5,
       cex.axis = 1.5, cex = 1, lwd = 3, lty = 5,
       xlab = expression(sigma), main = "",
       ylab = expression(pi(sigma ~"|"~ bold(y[n]))))
lines(c(1.565228, 3.015768), c(0.15, 0.15), col="darkorange",
       lwd = 3, lty = 5)
lines(density(sigma_post[,2], adjust=2), col = "darkgreen",
       lwd = 3, lty = 4)
lines(c(1.319387, 2.958068), c(0.12, 0.12), col = "darkgreen",
       lwd = 3, lty = 4)
lines(density(sigma_post[,3], adjust=2), col = "darkblue",
       lwd = 3)
lines(c(0.9618679, 2.6742025), c(0.14, 0.14),
       col = "darkblue", lwd = 3)

############# Analysis excluding the two outliers ############
############### First computation method: MCMC ###############
# The results in the paper were produced using this method.

set.seed(1)
sim_norm <- mcmc(nb_iter = 10 ^ 7, x[-c(17,20)], y[-c(17,20)],
    theta = 0.5, initial_val = c(0.34, 1.18), scaling_prop =
    0.10, law = rnorm, log_likely = log_likely_normal)

set.seed(1)
sim_student <- mcmc(nb_iter = 10 ^ 7, x[-c(17,20)],
  y[-c(17,20)], theta = 0.5, initial_val = c(0.34, 1.27),
  scaling_prop = 0.10, law = rt10, log_likely =
  log_likely_student)

set.seed(1)
sim_LPTN <- mcmc(nb_iter = 10 ^ 7, x[-c(17,20)], y[-c(17,20)],
  theta = 0.5, initial_val = c(0.34, 1.19), scaling_prop =
  0.10, law = rLPTN, log_likely = log_likely_LPTN)

print(c(sim_norm$rate, sim_student$rate, sim_LPTN$rate))
[1] 0.2212470 0.2537879 0.2224217

beta_post <- cbind(sim_norm$estim[,1], sim_student$estim[,1],
             sim_LPTN$estim[,1])
beta_median <- apply(beta_post, 2, median)
print(beta_median)
[1] 0.3420777 0.3389598 0.3427098

sigma_post <- cbind(sim_norm$estim[,2], sim_student$estim[,2],
              sim_LPTN$estim[,2])
sigma_median <- apply(sigma_post, 2, median)
print(sigma_median)
[1] 1.176685 1.267760 1.190079

######################## HPD intervals #######################

HPD_beta <- apply(beta_post, 2, function(gen){HPDinterval(
            as.mcmc(gen), prob = 0.95)[1:2]})
print(HPD_beta)
[1,] 0.302074 0.2980923 0.3033320
[2,] 0.381890 0.3803959 0.3817377

HPD_sigma <- apply(sigma_post, 2, function(gen){HPDinterval(
             as.mcmc(gen), prob = 0.95)[1:2]})
print(HPD_sigma)
[1,] 0.8246836 0.8499801 0.8540482
[2,] 1.6556265 1.8225438 1.6606780

########################## Figure 5a #########################

par(mar = c(4.5, 5, 1, 3))
plot(x[-c(17,20)], y[-c(17,20)], type = "p", pch = 19,
       cex.lab = 1.5, cex.axis = 1.5, cex = 1,
       xlab = "Income", ylab = "Expenditure on food",
       xlim = c(0, 700), ylim = c(0, 120))
abline(a = 0, b = beta_median[1], col="darkorange", lwd = 3,
       lty = 5)
abline(a = 0, b = beta_median[2], col="darkgreen", lwd = 3,
       lty = 4)
abline(a = 0, b = beta_median[3], col="darkblue", lwd = 3)

########################## Figure 5b #########################

par(mar = c(4.5, 6, 1, 3))
plot(density(beta_post[,1], adjust=2), xlim = c(0.15, 0.45),
       ylim = c(0, 22), col = "darkorange", cex.lab = 1.5,
       cex.axis = 1.5, cex = 1, lwd = 3, lty = 5,
       xlab = expression(beta), main = "",
       ylab = expression(pi(beta ~"|"~ bold(y[k]))))
lines(c(0.302074, 0.381890), c(2.57, 2.57), col="darkorange",
       lwd = 3, lty = 5)
lines(density(beta_post[,2], adjust=2), col = "darkgreen",
       lwd = 3, lty = 4)
lines(c(0.2980923, 0.3803959), c(2.53, 2.53),
       col = "darkgreen", lwd = 3, lty = 4)
lines(density(beta_post[,3], adjust=2), col = "darkblue",
       lwd = 3)
lines(c(0.3033320, 0.3817377), c(2.54, 2.54),
       col = "darkblue", lwd = 3)

########################## Figure 5c #########################

par(mar = c(4.5, 6, 1, 3))
plot(density(sigma_post[,1], adjust=2), xlim = c(0.5, 4.5),
       ylim = c(0, 2.2), col = "darkorange", cex.lab = 1.5,
       cex.axis = 1.5, cex = 1, lwd = 3, lty = 5,
       xlab = expression(sigma), main = "",
       ylab = expression(pi(sigma ~"|"~ bold(y[k]))))
lines(c(0.8246836, 1.6556265), c(0.23, 0.23),
       col= "darkorange", lwd = 3, lty = 5)
lines(density(sigma_post[,2], adjust=2), col = "darkgreen",
       lwd = 3, lty = 4)
lines(c(0.8499801, 1.8225438), c(0.20, 0.20),
       col = "darkgreen", lwd = 3, lty = 4)
lines(density(sigma_post[,3], adjust=2), col = "darkblue",
       lwd = 3)
lines(c(0.8540482, 1.6606780), c(0.25, 0.25),
       col = "darkblue", lwd = 3)

########### Second computation method: Riemann sums ##########
# This method can be used to validate the results produced by
# the first method, or simply when it is preferred by the
# user.

law_normal <- function(z) {dnorm(z, log = TRUE)}

law_student <- function(z) {
  df <- 10
  # known_scale added to match 2.5 & 97.5th perc. of N(0,1)
  known_scale <- round(qnorm(0.025) / qt(.025, df), 2)
  dt(z / known_scale, df = df, log = TRUE)
}

law_LPTN <- function(z) {
    alpha <- 1.96
    q <- 2 * pnorm(alpha) - 1
    phi <- 1 + 2 * dnorm(alpha) * alpha * log(alpha) / (1 - q)
    z_abs <- abs(z)
    tails <- as.numeric(z_abs <= alpha)
    # to avoid undefined value of log(log(z)) if tails=1
    z_floor <- apply(cbind(z_abs), 1, max, 1.001)
    logf <- tails * dnorm(z, log = TRUE) + (1 - tails) *
      (dnorm(alpha, log = TRUE) + log(alpha) - log(z_floor) +
       phi * log(log(alpha)) - phi * log(log(z_floor)))
    return(logf)
}

posterior_density <- function(beta1, beta2, dbeta, sigma1,
                              sigma2, dsigma, x, y, law){
   theta <- 0.5
   beta <- seq(beta1, beta2, dbeta)
   sigma <- seq(sigma1, sigma2, dsigma)
   log_posterior_beta <- function(sig){
        n <- length(x)
        res <- 0
        for (i in 1:n){
            z <- (y[i] - beta * x[i]) /
                 (abs(x[i]) ^ theta * sig)
            res <- res + law(z)
        }
        # prior is 1 / sigma
        return(res - (n + 1) * log(sig))
   }
   log_posterior <- apply(cbind(sigma), 1, log_posterior_beta)
   posterior_density_propto <- exp(log_posterior)
   normalizing_cte <- sum(posterior_density_propto) *
                      dbeta * dsigma
   return(posterior_density_propto / normalizing_cte)
}

# integration bounds
beta1 <- -0.3; beta2 <- 0.9; dbeta <- 0.001
beta_range <- seq(beta1, beta2, dbeta)
sigma1 <- 0.3; sigma2 <- 13.0; dsigma <- 0.001
sigma_range <- seq(sigma1, sigma2, dsigma)

############ posterior densities of beta and sigma ###########

post_normal <- posterior_density(beta1, beta2, dbeta, sigma1,
               sigma2, dsigma, x, y, law = law_normal)

post_student <- posterior_density(beta1, beta2, dbeta, sigma1,
                sigma2, dsigma, x, y, law = law_student)

post_LPTN <- posterior_density(beta1, beta2, dbeta, sigma1,
             sigma2, dsigma, x, y, law = law_LPTN)

######## posterior densities of beta for the 3 models ########

beta_post <- cbind(apply(post_normal, 1, sum),
                   apply(post_student, 1, sum),
                   apply(post_LPTN, 1, sum)) * dsigma

# Check if bounds are suitable, i.e. if the densities are
# low
beta_post[1,]; beta_post[nrow(beta_post),]
[1] 1.766528e-12 2.376120e-13 6.711534e-13
[1] 5.863097e-13 1.219712e-13 1.897941e-13

beta_median <- beta_range[apply(apply(beta_post, 2, cumsum) *
                          dbeta < .5, 2, sum)]
print(beta_median)
[1] 0.282 0.305 0.318

######## posterior densities of sigma for the 3 models #######

sigma_post <- cbind(apply(post_normal, 2, sum),
                    apply(post_student, 2, sum),
                    apply(post_LPTN, 2, sum)) * dbeta

# Check if bounds are suitable, i.e. if the densities are
# low
sigma_post[1,]; sigma_post[nrow(sigma_post),]
[1] 3.329440e-190  4.095789e-31  6.891520e-15
[1] 2.597696e-12   5.541500e-12  9.870326e-13

sigma_median <- sigma_range[apply(apply(sigma_post, 2,
                     cumsum) * dbeta < .5, 2, sum)]
print(sigma_median)
[1] 2.181 2.032 1.633

######################## HPD intervals #######################

HPD_unimodal <- function(z, dz, dens,
                         conf_level_target = 0.95){
  conf_level <- 0
  dens_value <- max(dens)
  while(conf_level < conf_level_target){
    dens_value <- dens_value - 0.0001
    pos <- which(dens > dens_value)
    conf_level <- sum(dens[pos]) * dz
  }
  return(c(z[pos[1]], z[pos[length(pos)]]))
}

# Computation of the 95 % HPD Intervals for the 3 models
HPD_unimodal(beta_range, dbeta, beta_post[,1], 0.95)
HPD_unimodal(beta_range, dbeta, beta_post[,2], 0.95)
HPD_unimodal(beta_range, dbeta, beta_post[,3], 0.95)
HPD_unimodal(sigma_range, dsigma, sigma_post[,1], 0.95)
HPD_unimodal(sigma_range, dsigma, sigma_post[,2], 0.95)
HPD_unimodal(sigma_range, dsigma, sigma_post[,3], 0.95)
[1] 0.218 0.349
[1] 0.243 0.366
[1] 0.240 0.376
[1] 1.560 3.007
[1] 1.32 2.96
[1] 0.961 2.671

########################## Figure 4a #########################

par(mar = c(4.5, 5, 1, 3))
plot(x, y, type = "p", pch = 19, cex.lab = 1.5,
       cex.axis = 1.5, cex = 1, xlab = "Income",
       ylab = "Expenditure on food", xlim = c(0, 700),
       ylim = c(0, 120))
abline(a = 0, b = beta_median[1], col="darkorange", lwd = 3,
       lty = 5)
abline(a = 0, b = beta_median[2], col="darkgreen", lwd = 3,
       lty = 4)
abline(a = 0, b = beta_median[3], col="darkblue", lwd = 3)

########################## Figure 4b #########################

par(mar = c(4.5, 6, 1, 3))
posi1 <- round((0.15 - beta1) / dbeta + 1)
posi2 <- round((0.45 - beta1) / dbeta + 1)
plot(beta_range[posi1:posi2], beta_post[posi1:posi2,1],
       type="l", xlim = c(0.15, 0.45),
       ylim = c(0, 14), col = "darkorange", cex.lab = 1.5,
       cex.axis = 1.5, cex = 1, lwd = 3, lty = 5,
       xlab = expression(beta), main = "",
       ylab = expression(pi(beta ~"|"~ bold(y[n]))))
lines(c(0.218, 0.349), c(1.5, 1.5), col="darkorange",
       lwd = 3, lty = 5)
lines(beta_range[posi1:posi2], beta_post[posi1:posi2,2],
      col="darkgreen", lwd = 3, lty = 4)
lines(c(0.243, 0.366), c(1.7, 1.7), col = "darkgreen",
       lwd = 3, lty = 4)
lines(beta_range[posi1:posi2], beta_post[posi1:posi2,3],
      col="darkblue", lwd = 3)
lines(c(0.240, 0.376), c(1.6, 1.6), col = "darkblue",
       lwd = 3)

########################## Figure 4c #########################

par(mar = c(4.5, 6, 1, 3))
posi1 <- round((0.5 - sigma1) / dsigma + 1)
posi2 <- round((4.5 - sigma1) / dsigma + 1)
plot(sigma_range[posi1:posi2], sigma_post[posi1:posi2,1],
       type="l", ylim = c(0, 1.3),
       col = "darkorange", cex.lab = 1.5,
       cex.axis = 1.5, cex = 1, lwd = 3, lty = 5,
       xlab = expression(beta), main = "",
       ylab = expression(pi(sigma ~"|"~ bold(y[n]))))
lines(c(1.560, 3.007), c(0.15, 0.15), col="darkorange",
       lwd = 3, lty = 5)
lines(sigma_range[posi1:posi2], sigma_post[posi1:posi2,2],
     type="l", col="darkgreen", lwd = 3, lty = 4)
lines(c(1.32, 2.96), c(0.12, 0.12), col = "darkgreen",
       lwd = 3, lty = 4)
lines(sigma_range[posi1:posi2], sigma_post[posi1:posi2,3],
     type="l", col="darkblue", lwd = 3)
lines(c(0.961, 2.671), c(0.14, 0.14),
       col = "darkblue", lwd = 3)

############# Analysis excluding the two outliers ############
########### Second computation method: Riemann sums ##########

# integration bounds
beta1 <- -0.3; beta2 <- 0.9; dbeta <- 0.001
beta_range <- seq(beta1, beta2, dbeta)
sigma1 <- 0.3; sigma2 <- 13.0; dsigma <- 0.001
sigma_range <- seq(sigma1, sigma2, dsigma)

############ posterior densities of beta and sigma ###########

post_normal <- posterior_density(beta1, beta2, dbeta, sigma1,
               sigma2, dsigma, x[-c(17,20)], y[-c(17,20)],
               law = law_normal)

post_student <- posterior_density(beta1, beta2, dbeta, sigma1,
               sigma2, dsigma, x[-c(17,20)], y[-c(17,20)],
               law = law_student)

post_LPTN <- posterior_density(beta1, beta2, dbeta, sigma1,
               sigma2, dsigma, x[-c(17,20)], y[-c(17,20)],
               law = law_LPTN)

######## posterior densities of beta for the 3 models ########

beta_post <- cbind(apply(post_normal, 1, sum),
                   apply(post_student, 1, sum),
                   apply(post_LPTN, 1, sum)) * dsigma

# Check if bounds are suitable, i.e. if the densities are
# low
beta_post[1,]; beta_post[nrow(beta_post),]
[1] 5.557673e-16 3.139864e-16 5.853859e-16
[1] 6.934071e-15 4.419829e-15 7.292775e-15

beta_median <- beta_range[apply(apply(beta_post, 2, cumsum) *
                          dbeta < .5, 2, sum)]
print(beta_median)
[1] 0.341 0.338 0.342

######## posterior densities of sigma for the 3 models #######

sigma_post <- cbind(apply(post_normal, 2, sum),
                    apply(post_student, 2, sum),
                    apply(post_LPTN, 2, sum)) * dbeta

# Check if bounds are suitable, i.e. if the densities are
# low
sigma_post[1,]; sigma_post[nrow(sigma_post),]
[1] 8.204966e-41 4.550658e-19 1.859492e-12
[1] 1.066419e-15 7.892179e-15 1.123373e-15

sigma_median <- sigma_range[apply(apply(sigma_post, 2,
                cumsum) * dbeta < .5, 2, sum)]
print(sigma_median)
[1] 1.176 1.267 1.189

######################## HPD intervals #######################

# Computation of the 95 % HPD Intervals for the 3 models
HPD_unimodal(beta_range, dbeta, beta_post[,1], 0.95)
HPD_unimodal(beta_range, dbeta, beta_post[,2], 0.95)
HPD_unimodal(beta_range, dbeta, beta_post[,3], 0.95)
HPD_unimodal(sigma_range, dsigma, sigma_post[,1], 0.95)
HPD_unimodal(sigma_range, dsigma, sigma_post[,2], 0.95)
HPD_unimodal(sigma_range, dsigma, sigma_post[,3], 0.95)
[1] 0.303 0.382
[1] 0.298 0.380
[1] 0.304 0.382
[1] 0.824 1.653
[1] 0.850 1.824
[1] 0.853 1.660

########################## Figure 5a #########################

par(mar = c(4.5, 5, 1, 3))
plot(x, y, type = "p", pch = 19, cex.lab = 1.5,
       cex.axis = 1.5, cex = 1, xlab = "Income",
       ylab = "Expenditure on food", xlim = c(0, 700),
       ylim = c(0, 120))
abline(a = 0, b = beta_median[1], col="darkorange", lwd = 3,
       lty = 5)
abline(a = 0, b = beta_median[2], col="darkgreen", lwd = 3,
       lty = 4)
abline(a = 0, b = beta_median[3], col="darkblue", lwd = 3)

########################## Figure 5b #########################

par(mar = c(4.5, 6, 1, 3))
posi1 <- round((0.15 - beta1) / dbeta + 1)
posi2 <- round((0.45 - beta1) / dbeta + 1)
plot(beta_range[posi1:posi2], beta_post[posi1:posi2,1],
       type="l", xlim = c(0.15, 0.45),
       ylim = c(0, 22), col = "darkorange", cex.lab = 1.5,
       cex.axis = 1.5, cex = 1, lwd = 3, lty = 5,
       xlab = expression(beta), main = "",
       ylab = expression(pi(beta ~"|"~ bold(y[n]))))
lines(c(0.303, 0.382), c(2.57, 2.57), col="darkorange",
       lwd = 3, lty = 5)
lines(beta_range[posi1:posi2], beta_post[posi1:posi2,2],
      col="darkgreen", lwd = 3, lty = 4)
lines(c(0.298, 0.380), c(2.53, 2.53), col = "darkgreen",
       lwd = 3, lty = 4)
lines(beta_range[posi1:posi2], beta_post[posi1:posi2,3],
      col="darkblue", lwd = 3)
lines(c(0.304, 0.382), c(2.54, 2.54), col = "darkblue",
       lwd = 3)

########################## Figure 5c #########################

par(mar = c(4.5, 6, 1, 3))
posi1 <- round((0.5 - sigma1) / dsigma + 1)
posi2 <- round((4.5 - sigma1) / dsigma + 1)
plot(sigma_range[posi1:posi2], sigma_post[posi1:posi2,1],
       type="l", ylim = c(0, 2.2),
       col = "darkorange", cex.lab = 1.5,
       cex.axis = 1.5, cex = 1, lwd = 3, lty = 5,
       xlab = expression(beta), main = "",
       ylab = expression(pi(sigma ~"|"~ bold(y[n]))))
lines(c(0.824, 1.653), c(0.23, 0.23), col="darkorange",
       lwd = 3, lty = 5)
lines(sigma_range[posi1:posi2], sigma_post[posi1:posi2,2],
     type="l", col="darkgreen", lwd = 3, lty = 4)
lines(c(0.850, 1.824), c(0.20, 0.20), col = "darkgreen",
       lwd = 3, lty = 4)
lines(sigma_range[posi1:posi2], sigma_post[posi1:posi2,3],
     type="l", col="darkblue", lwd = 3)
lines(c(0.853, 1.660), c(0.25, 0.25),
       col = "darkblue", lwd = 3)

######################### Section 3.2 ########################

library("MASS")
library("robustbase")

M_model <- function(x, y, theta){
  adj <- abs(x) ^ theta
  regM <- rlm(y / adj ~ 0 + I(x / adj), method= "M",
          maxit = 150)
  return(c(coefficients(regM),summary(regM)$sigma))
}

S_model <- function(x, y, theta){
  adj <- abs(x) ^ theta
  regS <- lmrob.S(y = y / adj, x = x / adj, control =
          lmrob.control(nRes = 20, k.max = 5000,
          max.it = 5000, maxit.scale = 500))
  return(c(coefficients(regS),regS$scale))
}

rnorm_mix_sd <- function(n){
    sds <- c(1.0, 10.0) # a vector containing the std. dev.
    components <- sample(1:2, prob = c(0.9, 0.1), size = n,
                  replace = TRUE)
    return(rnorm(n, mean = 0, sd = sds[components]))
}

rnorm_mix_mu <- function(n){
    mus <- c(0.0, 10.0) # a vector containing the locations
    components <- sample(1:2, prob = c(0.95, 0.05), size = n,
                  replace = TRUE)
    return(rnorm(n, mean = mus[components], 1.0))
}

mse_normal_MS_model <- function(nb_sets, theta, n, beta,
                       sigma, law, model, seed = 1){
  # used for normal, M and S models
  set.seed(seed); x <- seq(1, n); matrix_error <- c(0, 0)
  for(i in 1:nb_sets){	
    y <- beta * x + sigma * (abs(x) ^ theta) * law(n)
    estimates <- model(x, y, theta)
    matrix_error <- matrix_error +
                    (estimates - c(beta, sigma)) ^ 2}
    return(matrix_error / nb_sets)}

mse_student_model <- function(nb_sets, theta, df, n, beta,
                     sigma, law, seed = 1){
  set.seed(seed); x <- seq(1, n); matrix_error <- c(0, 0)
  for(i in 1:nb_sets){	
    y <- beta * x + sigma * (abs(x) ^ theta) * law(n)
    estimates <- student_model(x, y, theta, df, beta, sigma)
    matrix_error <- matrix_error +
                    (estimates - c(beta, sigma)) ^ 2}
    return(matrix_error / nb_sets)}

mse_LPTN_model <- function(nb_sets, theta, alpha, n, beta,
                     sigma, law, seed = 1){
  set.seed(seed); x <- seq(1, n); matrix_error <- c(0, 0)
  for(i in 1:nb_sets){	
    y <- beta * x + sigma * (abs(x) ^ theta) * law(n)
    estimates <- LPTN_model(x, y, theta, alpha, beta, sigma)
    matrix_error <- matrix_error +
                    (estimates - c(beta, sigma)) ^ 2}
    return(matrix_error / nb_sets)}

nb_sets <- 10 ^ 6
r11 <- mse_normal_MS_model(nb_sets, 0.5, 20, 1.0, 1.5,
      rnorm, normal_model, 1)
r12 <- mse_normal_MS_model(nb_sets, 0.5, 20, 1.0, 1.5,
      rnorm_mix_sd, normal_model, 1)
r13 <- mse_normal_MS_model(nb_sets, 0.5, 20, 1.0, 1.5,
      rnorm_mix_mu, normal_model, 1)
r21 <- mse_student_model(nb_sets, 0.5, 10, 20, 1.0, 1.5,
       rnorm, 1)
r22 <- mse_student_model(nb_sets, 0.5, 10, 20, 1.0, 1.5,
       rnorm_mix_sd, 1)
r23 <- mse_student_model(nb_sets, 0.5, 10, 20, 1.0, 1.5,
       rnorm_mix_mu, 1)
r31 <- mse_LPTN_model(nb_sets, 0.5, 1.96, 20, 1.0, 1.5,
       rnorm, 1)
r32 <- mse_LPTN_model(nb_sets, 0.5, 1.96, 20, 1.0, 1.5,
       rnorm_mix_sd, 1)
r33 <- mse_LPTN_model(nb_sets, 0.5, 1.96, 20, 1.0, 1.5,
       rnorm_mix_mu, 1)
r41 <- mse_LPTN_model(nb_sets, 0.5, 1.5, 20, 1.0, 1.5,
       rnorm, 1)
r42 <- mse_LPTN_model(nb_sets, 0.5, 1.5, 20, 1.0, 1.5,
       rnorm_mix_sd, 1)
r43 <- mse_LPTN_model(nb_sets, 0.5, 1.5, 20, 1.0, 1.5,
       rnorm_mix_mu, 1)
r51 <- mse_normal_MS_model(nb_sets, 0.5, 20, 1.0, 1.5,
      rnorm, M_model, 1)
r52 <- mse_normal_MS_model(nb_sets, 0.5, 20, 1.0, 1.5,
      rnorm_mix_sd, M_model, 1)
r53 <- mse_normal_MS_model(nb_sets, 0.5, 20, 1.0, 1.5,
      rnorm_mix_mu, M_model, 1)
r61 <- mse_normal_MS_model(nb_sets, 0.5, 20, 1.0, 1.5,
      rnorm, S_model, 1)
r62 <- mse_normal_MS_model(nb_sets, 0.5, 20, 1.0, 1.5,
      rnorm_mix_sd, S_model, 1)
r63 <- mse_normal_MS_model(nb_sets, 0.5, 20, 1.0, 1.5,
      rnorm_mix_mu, S_model, 1)

j <- 1
table3 <- matrix(c(r11[j], r12[j], r13[j], r21[j], r22[j],
          r23[j], r31[j], r32[j], r33[j], r41[j], r42[j],
          r43[j], r51[j], r52[j], r53[j], r61[j], r62[j],
          r63[j]), byrow = TRUE, nrow = 6)

table3
[1,] 0.01069591 0.11655420 0.10978280
[2,] 0.01099606 0.02741262 0.03307914
[3,] 0.01115355 0.01990043 0.01756001
[4,] 0.01258564 0.01608100 0.01347029
[5,] 0.01127038 0.01705433 0.01732239
[6,] 0.02917979 0.02702641 0.02739469

round(table3, digits = 3)
[1,] 0.011 0.117 0.110
[2,] 0.011 0.027 0.033
[3,] 0.011 0.020 0.018
[4,] 0.013 0.016 0.013
[5,] 0.011 0.017 0.017
[6,] 0.029 0.027 0.027

j <- 2
table4 <- matrix(c(r11[j], r12[j], r13[j], r21[j], r22[j],
          r23[j], r31[j], r32[j], r33[j], r41[j], r42[j],
          r43[j], r51[j], r52[j], r53[j], r61[j], r62[j],
          r63[j]), byrow = TRUE, nrow = 6)

table4
[1,] 0.05875654 12.9842653 5.0272968
[2,] 0.06386030  4.0156380 1.9934910
[3,] 0.06727904  0.6044520 0.2220696
[4,] 0.09011541  0.1969683 0.1111314
[5,] 0.14246930  0.2415919 0.2064951
[6,] 0.10685390  0.2285541 0.1546262

round(table4, digits = 2)

[1,] 0.06 12.98 5.03
[2,] 0.06  4.02 1.99
[3,] 0.07  0.60 0.22
[4,] 0.09  0.20 0.11
[5,] 0.14  0.24 0.21
[6,] 0.11  0.23 0.15

The computation time for the normal model is negligible
(a few seconds). The second fastest approach is the S-model
with approximately 400 seconds for 1,000,000 data sets and
19 seconds for 50,000 data set using a regular laptop. In
comparison with this model, the Student model takes
approximately 1.5 times longer, the M-model takes 5 times
longer, the LPTN model with alpha = 1.96 takes 10 times
longer and the LPTN model with alpha = 1.5 takes 14 times
longer.
\end{verbatim}

\end{document}